\documentclass{article}
\usepackage[T1]{fontenc}
\usepackage{amsfonts,amsmath,amsthm,amssymb,authblk,latexsym,bm}
\usepackage{hyperref} \hypersetup{colorlinks=true}
\usepackage{citeref}
\usepackage{color}

\newtheorem{theorem}{Theorem}[section]
\newtheorem{corollary}[theorem]{Corollary}

\newtheorem{lemma}[theorem]{Lemma}
\newtheorem{definition}[theorem]{Definition}
\newtheorem{example}[theorem]{Example}

\newtheorem{remark}[theorem]{Remark}

\numberwithin{equation}{section}

\def\vph{\varphi} 
\def\dash{\mbox{-}}
\def\R1{{\cal R}}
\def\RL{{\cal R}_{\kern-1pt\lambda}}

\begin{document}

\title{Galois Self-dual 2-quasi Constacyclic Codes over Finite Fields}
\author{
Yun Fan, ~ Yue Leng\par
{\small School of Mathematics and Statistics}\par\vskip-1mm
{\small Central China Normal University, Wuhan 430079, China}

%{\small School of Mathematics and Statistics}\par\vskip-1mm
%{\small Central China Normal University, Wuhan 430079, China}
}

\date{}
\maketitle

\insert\footins{\footnotesize{\it Email address}:
yfan@mail.ccnu.edu.cn (Yun Fan); yueleng0213@126.com (Yue Leng).
}

\begin{abstract}
Let $F$ be a field with cardinality $p^\ell$ and $0\neq \lambda\in F$, and
$0\le h<\ell$.
Extending Euclidean and Hermitian inner products,
Fan and Zhang introduced Galois $p^h$-inner product
(DCC, vol.84, pp.473-492).
% and studied the existence of Galois self-dual constacyclic codes.
In this paper, we characterize the structure
of $2$-quasi $\lambda$-constacyclic codes over~$F$;
and exhibit necessary and sufficient conditions
 for  $2$-quasi $\lambda$-constacyclic codes being Galois $p^h$-self-dual.
With the help of a technique developed in this paper,
we prove that, when $\ell$ is even,
the Hermitian self-dual $2$-quasi $\lambda$-constacyclic codes
are asymptotically good if and only if $\lambda^{1+p^{\ell/2}}\!=1$.
And, when $p^\ell\,{\not\equiv}\,3~({\rm mod}~4)$,
the Euclidean self-dual $2$-quasi $\lambda$-constacyclic codes
are asymptotically good if and only if $\lambda^{2}=1$.

\end{abstract}

\section{Introduction}

Let $F$ be a finite field with cardinality $|F|=q=p^\ell$
where $p$ is a prime and $\ell$ is a positive integer,
let $F^\times=F\setminus\{0\}$.
Any ${\bf a}=(a_0, a_1, \dots, a_{n-1})\in{F}^n$
is called a word over $F$ of length $n$, where $n$ is a positive integer.
The Hamming weight ${\rm w}({\bf a})$ is defined as the number
of the indexes $i$ with $a_i\ne 0$.
The Hamming distance between two words ${\bf a}, {\bf a}'\in F^n$
is defined as ${\rm d}({\bf a},{\bf a}')={\rm w}({\bf a}-{\bf a}')$.
Any nonempty subset $C\subseteq F^{n}$ is called a code of length $n$ over $F$,
and the words in the code are called codewords.
The {\em minimum distance}
${\rm d}(C):=\min\limits_{{\bf c}\ne{\bf c}'\in C}{\rm d}({\bf c}, {\bf c}')$.
For any linear subspace $C$ of $F^n$, called a {\em linear code},
 the {\em minimum weight}
${\rm w}(C):=\min\limits_{0\ne{\bf c}\in C}{\rm w}({\bf c})$;
and it is known that ${\rm w}(C)={\rm d}(C)$.
The fraction $\Delta(C)=\frac{{\rm d}(C)}{n}=\frac{{\rm w}(C)}{n}$
is called the {\em relative minimum distance} of~$C$,
and ${\rm R}(C)=\frac{\dim_{F}C}{n}$ is called the {\em rate} of $C$.
A code sequence $C_{1}, C_{2}, \dots$ is said to be {\em asymptotically good}
if the length~$n_{i}$ of~$C_{i}$ goes to infinity and there is a real number $\delta>0$
such that ${\rm R}(C_{i})>\delta$ and $\Delta(C_{i})>\delta$
for $i=1,2,\dots$.
A class of codes is said to be {\em asymptotically good} if
there is an asymptotically good sequence of codes in the class;
otherwise, we say that the class of codes is {\em asymptotically bad}.

A linear code $C$ of $F^n$ is called a {\em cyclic code}
if $C$ is invariant under the cyclic permutation on items, i.e.,
\begin{align*} %\label{int def cyclic}
(c_0, c_1, \ldots, c_{n-1})\in{C} \implies
( c_{n-1}, c_0, c_1, \dots, c_{n-2})\in{C}.
\end{align*}
A linear code $C$ of $F^n\times F^n$ is called a {\em quasi-cyclic code of index $2$},
abbreviated as {\em $2$-quasi-cyclic code},
if $C$ is invariant under the double cyclic permutation on items, i.e.,
\begin{align*} %\label{int def 2-cyclic}
\begin{array}{l}
(c_0, c_1, \dots, c_{n-1}, c'_0, c'_1, \dots, c'_{n-1})\in{C} \\[1pt]
\implies (c_{n-1},c_0,c_1,\dots,c_{n-2}, \,
  c'_{n-1},c'_0,c'_1,\dots, c'_{n-2})\in{C}.
\end{array}
\end{align*}

The Euclidean inner product of words
${\bf a}=(a_0,\dots,a_{n-1}), {\bf b}=(b_0,\dots,b_{n-1})$ of $F^n$
 is defined to be $\langle{\bf a},{\bf b}\rangle=\sum_{i=0}^{n-1}a_ib_i$.
For a code $C\subseteq F^n$, the
$C^{\bot}=\{\,{\bf a}\in{F^{n}\,|\,\langle {\bf c},{\bf a}\rangle=0,\, \forall {\bf c}\in{C}}\,\}$
is the {\em dual code} of $C$.
A code $C$ is said to be {\em self-dual} if $C=C^{\bot}$.
Obviously, the rate $R(C)=\frac{1}{2}$ if $C$ is self-dual.

Cyclic codes are investigated extensively in theory and practice, cf. \cite{HP}.
It is still an open question (cf. \cite {MW06}):
are cyclic codes over $F$ asymptotically good?
However, it is well-known long ago that
the binary $2$-quasi-cyclic codes are asymptotically good, see \cite{CPW,C,K}.
Later, Mart\'inez-P\'erez and Willems \cite{MW} proved the
asymptotic goodness of binary self-dual $2$-quasi-cyclic codes.
And, \cite{AOS}, \cite{L PhD} and \cite{LF22} proved
that, if $q\,{\not\equiv}\, 3~({\rm mod}~4)$,
the $q$-ary self-dual $2$-quasi-cyclic codes are asymptotically good.
Note that ``$q\,{\not\equiv}\, 3~({\rm mod}~4)$'' is a necessary and sufficient
condition for the existence of $q$-ary self-dual $2$-quasi-cyclic codes,
cf. \cite[Theorem 6.1]{LS03}.
The proof in \cite{AOS} is based on Artin's primitive root conjecture.
The arguments in \cite{L PhD} and \cite{LF22} are self-contained.
And the asymptotic goodness of any $q$-ary $2$-quasi-cyclic codes
were also proved in~\cite{L PhD}.

Cyclic codes and $2$-quasi-cyclic codes had been extended widely.
Let $\lambda\in F^\times=F\setminus\{0\}$.
A linear code $C$ of $F^n$ is called a {\em $\lambda$-constacyclic code}
if $C$ is invariant under the $\lambda$-constacyclic permutation on items, i.e.,
\begin{align} \label{int def lambda-cyclic}
(c_0, c_1, \ldots, c_{n-1})\in C \implies
(\lambda c_{n-1}, c_0, c_1, \dots, c_{n-2})\in C.
\end{align}
If $\lambda=-1$, the $\lambda$-constacyclic codes are called
{\em negacyclic codes}. Further,
a linear code $C$ of $F^n\times F^n$ is called
a {\em $2$-quasi $\lambda$-constacyclic code} if
$C$ is invariant under the double $\lambda$-constacyclic permutation on items, i.e.,
\begin{align}\label{int def 2 lambda-cyclic}
\begin{array}{l}
(c_0, c_1, \dots, c_{n-1}, c'_0, c'_1, \dots, c'_{n-1})\in{C} \\[1pt]
\implies (\lambda c_{n-1},c_0,c_1,\dots,c_{n-2}, \,
 \lambda c'_{n-1},c'_0,c'_1,\dots, c'_{n-2})\in{C}.
\end{array}
\end{align}

If $q$ is odd and $q\,{\not\equiv}\pm 1~({\rm mod}~8)$,
the self-dual $2$-quasi negacyclic codes over~$F$ are proved asymptotically good
in \cite{AGOSS}.
While in \cite{SQS} for $q\,{\equiv}-1~({\rm mod}~4)$
it is shown, based on Artin's primitive root conjecture,
that the $q$-ary self-dual $2$-quasi negacyclic codes are asymptotically good.
Recently, for any $q$ and any $\lambda\in F^\times$
the $2$-quasi $\lambda$-constacyclic codes over $F$ are proved
asymptotically good, see \cite[Corollary I.3]{FL22}.

About the self-dualities, in the semisimple case (i.e., $\gcd(n,q)=1$),
the self-dual cyclic codes over $F$ does not exist.
Leon et al. \cite{LMP} and many references, e.g. \cite{DP, LX, P87},
 devoted to the study on various generalizations,
e.g., duadic codes, extended self-dual cyclic codes, etc.
On the other hand,
Dinh and Lopez-Permouth \cite{DL}, Dinh \cite{D12}
studied $\lambda$-constacyclic codes, and showed that in the semisimple case
the self-dual $\lambda$-constacyclic codes exist only if $\lambda=-1$.

Extending  the Euclidean inner product and the Hermitian inner product,
Fan and Zhang \cite{FZ17} introduced the so-called
Galois inner products. %(or Galois inner product),
Recall that $|F|=q=p^\ell$. For $0\le h<\ell$,
the map $\sigma_{p^h}: F\to F$, $\alpha\mapsto\alpha^{p^h}$,
is a Galois automorphism of $F$, which induces an automorphism
$F^n\to F^n$, $(a_0,\dots,a_{n-1})\mapsto(a_0^{p^h},\dots,a_{n-1}^{p^h})$.
The following
\begin{align} \label{int Galois inner}
\langle{\bf a},{\bf b}\rangle_{h}
 =\sum_{i=0}^{n-1}a_ib_i^{p^h},
~~~~ \forall\, {\bf a}=(a_0,\cdots,a_{n-1}),{\bf b}=(b_0,\cdots,b_{n-1})\in F^n,
\end{align}
is called the {\em Galois $p^h$-inner product} on $F^n$.
And for any code $C\subseteq F^n$, the following code
\begin{align} \label{int Galois dual}
C^{\bot{h}}=\big\{\,{\bf a}\in F^n\,\big|\,
  \langle {\bf c},{\bf a}\rangle_h=0,\; \forall\, {\bf c}\in C\big\}
\end{align}
is called the {\em Galois $p^h$-dual code} of $C$.
The code $C$ is said to be {\em Galois $p^h$-self-dual}
(or {\em Galois self-dual} when $h$ is known from context) if $C=C^{\bot h}$.
It is also obvious that $R(C)=\frac{1}{2}$ if $C$ is Galois $p^h$-self-dual.
When $h=0$ ($h=\ell/2$ when $\ell$ is even, respectively),
 $\langle{\bf a},{\bf b}\rangle_{h}$ is just the Euclidean
(Hermitian, respectively) inner product, and
$C^{\bot{h}}$ is the Euclidean (Hermitian, respectively) dual code,
and Galois $p^h$-self-dual codes are just the Euclidean self-dual
(Hermitian self-dual, respectively) codes.
The existence and the structure
of Galois $p^h$-self-dual $\lambda$-constacyclic codes are studied in~\cite{FZ17}.

In this paper we study the Galois $p^h$-self-dual $2$-quasi $\lambda$-constacyclic codes
over $F$ and their asymptotic properties.
The main contributions of this paper are the following.
\begin{itemize}
\item
We characterize the algebraic structure of
the $2$-quasi $\lambda$-constacyclic codes
 and their Galois $p^h$-dual codes.
We find that the Galois $p^h$-self-dual $2$-quasi $\lambda$-constacyclic codes
behave very differently depending on whether ${\lambda^{1+p^h}=1}$ or not.
In both the cases we obtain necessary and sufficient conditions
for $2$-quasi $\lambda$-constacyclic codes being Galois $p^h$-self-dual.

\item
We obtain that, if $\lambda^{1+p^h}\neq 1$,
then the Galois $p^h$-self-dual
$2$-quasi $\lambda$-constacyclic codes are asymptotically bad.
On the other hand,
if $\ell$ is even and $\lambda^{1+p^{\ell/2}}=1$, then
the Hermitian self-dual $2$-quasi $\lambda$-constacyclic codes
are asymptotically good.
And, if $p^\ell\,{\not\equiv}\,3~({\rm mod}~4)$ and $\lambda^{2}=1$,
then the Euclidean self-dual $2$-quasi $\lambda$-constacyclic codes
are asymptotically good.
For the Euclidean case, we note that
the asymptotic goodness of the self-dual $2$-quasi-cyclic codes
has been proved in \cite{LF22};
on the other hand, for the asymptotic properties of the
self-dual $2$-quasi negacyclic codes, our result and the results in
\cite{AGOSS, SQS}, %are not mutually inclusive.
the three results do not cover each other.

\item
As for methodology, the so-called {\em reciprocal polynomial}
is a powerful tool for studying the duality property of
 $\lambda$-constacyclic and $2$-quasi $\lambda$-consta-cyclic codes,
 e.g., in \cite{AGOSS, SQS}.
It is revised in \cite{LF22} etc. to the  ``bar'' map of
the quotient ring $F[X]/\langle X^n-1\rangle$,
where $F[X]$ denotes the polynomial ring over $F$ and $\langle X^n-1\rangle$
denotes the ideal generated by $X^n-1$; cf. \cite[Remark 6.6(2)]{FZ23}.
For any matrix $A=(a_{ij})_{m\times n}$ over $F$,
the {\em Galois $p^h$-transpose} of~$A$ is defined to be
$A^{*h}=\big((a_{ij}^{p^h}\big)_{m\times n}\big)^T
=(a_{ji}^{p^h}\big)_{n\times m}$, cf. Eq.\eqref{eq def *h} below.
With the operator ``$*h$'' on matrices,
we introduce an operator ``$*$'' on the quotient ring
$F[X]/\langle X^n-\lambda\rangle$, $a(X)\mapsto a^*(X)$,
cf. Lemma~\ref{lem tau and *} below for details.
That operator becomes an useful technique
for studying the Galois duality property of $\lambda$-constacyclic codes.
That is a methodological innovation of the paper.
\end{itemize}

In Section~\ref {preliminaries}, some preliminaries are sketched.

In Section~\ref{2-Q CC}
we characterize the algebraic structure of the
 $2$-quasi $\lambda$-consta-cyclic codes over the finite field $F$
 and their Galois $p^h$-dual codes.

In Section~\ref{Section Galois duality} we study the
Galois $p^h$-self-dual $2$-quasi $\lambda$-constacyclic codes over $F$.
Our discussion divide into two cases: %consists of two parts according to that
 $\lambda^{1+p^h}\!\neq 1$ %$\lambda\neq\lambda^{-p^{\ell-h}}$
or $\lambda^{1+p^h}\!=1$. % $\lambda=\lambda^{-p^{\ell-h}}$.
In  both the cases,
we exhibit the necessary and sufficient conditions for a
$2$-quasi $\lambda$-constacyclic code being Galois $p^h$-self-dual.
And we show that if $\lambda^{1+p^h}\!\neq 1$,
%$\lambda\neq\lambda^{-p^{\ell-h}}$,
then Galois $p^h$-self-dual $2$-quasi $\lambda$-constacyclic codes
are asymptotically bad.

In Section~\ref{section H self-dual 2-Q cyclic},
the Hermitian self-dual $2$-quasi-cyclic codes over $F$ are proved asymptotically good.

In Section~\ref{section H self-dual 2-QQ cyclic},
assuming that $\ell$ is even, we prove that
the Hermitian self-dual $2$-quasi $\lambda$-constacyclic codes over $F$
are asymptotically good if and only if $\lambda^{1+p^{\ell/2}}=1$.
And, assuming that $p^\ell\,{\not\equiv}\,3~({\rm mod}~4)$,
we show that the Euclidean self-dual $2$-quasi $\lambda$-constacyclic codes over $F$
are asymptotically good if and only if $\lambda^{2}=1$.

Finally, %conclusion is made in Section~\ref{Conclusions}.
we end this paper by a conclusion in Section~\ref{Conclusions}.

\section{Preliminaries}\label{preliminaries}

In this paper, $F$ is always a finite field of cardinality $|F|=q=p^\ell$
(by $|S|$ we denote the cardinality of any set $S$),
where $p$ is a prime and $\ell$ is a positive integer;
and $h$ is an integer such that $0\le h<\ell$; and $n>1$ is an integer.
Any ring $R$ in this paper has identity $1_R$ (or denoted by $1$ for short);
and ring homomorphisms and  subrings are identity preserving.
By $R^\times$ we denote the multiplication group consisting
of all units (invertible elements) of $R$.
In particular, $F^\times=F\,{\setminus}\,\{0\}$.
If a ring $R$ is also an $F$-vector space,
then $R$ is said to be an $F$-algebra.
In that case, $F\to R$, $\alpha\mapsto \alpha 1_R$,
is an embedding, so that %with the embedding
we write that $F\subseteq R$.

Let $R$ and $S$ be $F$-algebras.
A map $\psi: R\to S$ is called an {\em $F$-algebra homomorphism}
if it is both a ring homomorphism and an $F$-linear map,
i.e., $\psi(1_R)=1_S$ and for any $a,b\in R$ and any $\alpha\in F$,
$$
\psi(ab)=\psi(a)\psi(b), ~~  \psi(a+b)=\psi(a)+\psi(b), ~~
 \psi(\alpha a)=\alpha\psi(a),
$$
where the first two mean that $\psi$ is a ring homomorphism,
and the last two mean that it is a linear map.
Recall that for $0\leq h < \ell $, the map $\sigma_{p^h}: F\to F$,
$\sigma_{p^h}(\alpha)=\alpha^{p^h}$ for
$\alpha\in F$, is a Galois automorphism of the field $F$.
If $\psi: R\to S$ is bijective and satisfies that
for any $a, b\in R$ and any $\alpha\in F$,
\begin{align} \label{eq alg iso}
   \psi(ab)=\psi(a)\psi(b), ~~ \psi(a+b)=\psi(a)+\psi(b), ~~
  \psi(\alpha a)=\alpha^{p^h}\psi(a),
\end{align}
then~$\psi$ is called a {\em $\sigma_{p^h}$-algebra isomorphism},
or {\em $p^h$-isomorphism} for short.
Note that the last two equalities of Eq.\eqref{eq alg iso}
mean that $\psi$ is a {\em $p^h$-linear map}.
In particular, if $R=S$,
then~$\psi$ is called a {\em $\sigma_{p^h}$-algebra automorphism},
or a {\em $p^h$-automorphism} for short.
And, if $\psi: R\to S$ is a $p^h$-isomorphism,
then for any ideal $I$ of $R$, %we have
the image $\psi(I)$ is an ideal of $S$ and
$\dim_F \psi(I)=\dim_F I$.
There are two typical examples as follows.

\begin{example}\label{exm} \rm
{\rm (1).} For any polynomial $f(X)=\sum_{i}a_iX^i\in F[X]$,
we denote $f^{(p^h)}(X)=\sum_{i} a_i^{p^h}X^i$.
The Galois automorphism $\sigma_{p^h}: F\to F$, $\alpha\mapsto\alpha^{p^h}$,
induces the map
\begin{align} \label{eq tilde sigma}
\tilde\sigma_{p^h}\!: ~ F[X]\,\longrightarrow\, F[X], ~~~~
 f(X)\,\longmapsto\,f^{(p^h)}(X),
\end{align}
which is a $p^h$-automorphism of $F[X]$.
Let ${\rm ord}(\sigma_{p^h})$ denote the order of %the automorphism
$\sigma_{p^h}$.  It is easy to check that
%${\rm ord}(\sigma_{p^h})=\frac{\ell}{\gcd(h,\ell)}$.%the order
${\rm ord}(\tilde\sigma_{p^h})
 ={\rm ord}(\sigma_{p^h})=\frac{\ell}{\gcd(h,\ell)}$.

{\rm (2).} For the matrix algebra
${\rm M}_n(F)=\big\{ (a_{ij})_{n\times n} \,\big|\, a_{ij}\in F\big\}$
consisting of all $n\times n$ matrices over $F$,
the map
\begin{align} \label{eq dot sigma}
\dot\sigma_{p^h}\!: {\rm M}_n(F)\to {\rm M}_n(F), ~~~~
\big(a_{ij}\big)_{n\times n}
  \longmapsto\big(a_{ij}^{p^h}\big)_{n\times n},
 %~~ \mbox{for}~\sum\limits_{i}a_i X^i \in F[X],
\end{align}
is a $p^h$-automorphism of ${\rm M}_n(F)$.
\end{example}

Let $\lambda\in F^\times$.
%The $\lambda$-constacyclic codes and
%the $2$-quasi $\lambda$-constacyclic codes are defined in
%Eq.\eqref{int def lambda-cyclic} and Eq.\eqref{int def 2 lambda-cyclic}.
%We show another description of them.
For a positive integer $k$, we write $E_k$ to denote the identity matrix of degree $k$.
Let $P_\lambda$ denote
the {\em $\lambda$-constacyclic permutation matrix} of degree $n$ as follows
\begin{align} \label{eq def P_lambda}
P_{\lambda}=\begin{pmatrix} & E_{n-1}\\ \lambda E_1\end{pmatrix}
=\begin{pmatrix} 0 & 1\\ & 0 & \ddots\\  && \ddots & \ddots\\
 & & & \ddots & 1 \\ \lambda &&&&0\end{pmatrix}.
\end{align}
In particular, if $\lambda=1$ then
$P:=P_1=\begin{pmatrix} & E_{n-1}\\ E_1\end{pmatrix}$
is the {\em cyclic permutation matrix}. By matrix multiplication,
for $\mathbf{a}=(a_0,a_1,\dots,a_{n-1})\in{F^n}$ we have that
\begin{align} \label{eq a dot P}
\mathbf{a}\cdot P_\lambda=
(a_0,a_1,a_2,\dots,a_{n-1})\cdot P_{\!\lambda}
=(\lambda a_{n-1},a_0,a_1,\dots,a_{n-2}),
\end{align}
which is the vector obtained by $\lambda$-constacyclically permuting
the items of the vector $(a_0,a_1,\dots,a_{n-1})$;  and that
for $\mathbf{a}=(a_0,a_1,\dots,a_{n-1})\in{F^n}$ and
$\mathbf{a'}=(a'_0,a'_1,\dots,a'_{n-1}) \in{F^n}$,
\begin{align} \label{eq aa dot PP}
\big(\mathbf{a},\,\mathbf{a}'\big)
\begin{pmatrix}P_\lambda\!\\ &\! P_\lambda  \end{pmatrix}
%&=\big(\mathbf{c}\cdot P_{\!\lambda},\,\mathbf{c}'\cdot P_{\!\lambda}\big)\\
&=(\lambda a_{n-1},a_0,a_1,\dots,a_{n-2}, \,
 \lambda a'_{n-1},a'_0,a'_1,\dots, a'_{n-2}).
\end{align}
Thus, we get another description of the $\lambda$-constacyclic codes
(cf. Eq.\eqref{int def lambda-cyclic}) and
the $2$-quasi $\lambda$-constacyclic codes (cf. Eq.\eqref{int def 2 lambda-cyclic})
 as follows.

\begin{lemma} \label{lambda-CC code}
{\bf(1)} Let $C$ be a subspace of $F^n$.
Then $C$ is a $\lambda$-constacyclic code if and only if
$\mathbf{c}\cdot P_\lambda \in C$, for any $\mathbf{c}\in C$.

{\bf(2)} Let $C$ be a subspace of $F^n\times F^{n}$.
Then $C$  is a $2$-quasi $\lambda$-constacyclic code
if and only if
$\Big(\mathbf{c},\,\mathbf{c}'\big)
\begin{pmatrix}P_\lambda\!\\ &\! P_\lambda  \end{pmatrix}
\in C$, for any $\big(\mathbf{c},\,\mathbf{c}'\big)\in C$. \qed
\end{lemma}

In the following we always denote
$\RL=F[X]\big/\langle X^n - \lambda\rangle$,
which is the quotient algebra of the polynomial algebra $F[X]$ over the ideal
$\langle X^n-\lambda\rangle$ generated by $X^n-\lambda$.
Any residue class modulo $X^n-\lambda$ has a unique representative
polynomial with degree less than $n$. Hence we can write %it as
\begin{align} \label{eq R_lambda}
\RL=F[X]\big/\langle X^n - \lambda\rangle
 =\big\{\,  a_0+a_1X+\dots+a_{n-1}X^{n-1}\;\big|\; a_i\in F\,\big\}.
\end{align}
Further, the Cartesian product
\begin{align}
 \RL^2=\RL\times\RL
 =\big\{\big(a(X),a'(X)\big)\;\big|\; a(X),~ a'(X)\in\RL \big\}
\end{align}
is an $\RL$-module.
For $\RL$ and $\RL^2$,
the following identifications and results will be quoted later in this article.

\begin{remark} \label{rk R ident. Fn} \rm
(1). There is a canonical linear isomorphism
\begin{align} \label{eq R to F}
\iota:~~  {\RL}~
\mathop{\longrightarrow}^{\cong}~ F^n,\qquad
a(X) ~\longmapsto~ \mathbf{a},
\end{align}
where $a(X)=a_0+a_1X+\dots+a_{n-1}X^{n-1}\in\RL$
and $\mathbf{a}=(a_0,a_1,\dots,a_{n-1})\in F^n$. It is easy to check that
\begin{align} \label{eq R cong F}
 \iota\big(X a(X)\big) = {\bf a}\cdot P_{\lambda},\qquad
\forall \;a(X)\in\RL.
\end{align}
Then any element $a(X)$ of $\RL$
is identified with the word $\mathbf{a}=(a_0,a_1,\dots,a_{n-1})$ of $F^n$;
and by Lemma~\ref{lambda-CC code}(1),
the $\lambda$-constacyclic codes of length $n$
are identified with the ideals ($\RL$-submodules) of $\RL$.

(2). For the $\RL$-module $\RL^2=\RL\times\RL$,
we have the following canonical linear isomorphism
\begin{align} \label{eq RR to FF}
\iota^{(2)}:~~  {\RL}\times{\RL}~
\mathop{\longrightarrow}^{\cong}~ F^n\times F^n,\qquad
\big(a(X),\,a'(X)\big) ~\longmapsto~ \big(\mathbf{a},\,\mathbf{a}'\big),
\end{align}
where $a(X)=a_0+a_1X+\dots+a_{n-1}X^{n-1}$,
$a'(X)=a'_0+a'_1X+\dots+a'_{n-1}X^{n-1}$,
and $(\mathbf{a},\mathbf{a}')=(a_0,a_1,\dots,a_{n-1},\, a'_0,a'_1,\dots,a'_{n-1})$.
For $\big(a(X),a'(X)\big)\in\RL^2$,
\begin{align} \label{eq RR cong FF}
 \iota^{(2)}\big(X\big(a(X),a'(X)\big)\big)
 =\iota^{(2)}\big(Xa(X),Xa'(X)\big)= \big({\bf a},{\bf a}'\big)
\begin{pmatrix}P_\lambda\!\\ &\! P_\lambda  \end{pmatrix}. %\in F^n\times F^n.
\end{align}
Then any element $\big(a(X),a'(X)\big)\in{{\RL}\times{\RL}}$
is identified with the word $\big(\mathbf{a},\,\mathbf{a}'\big)\in{F^n\times F^n}$,
and by Lemma~\ref{lambda-CC code}(2)
the $2$-quasi $\lambda$-constacyclic codes of length $2n$ are identified
with the ${\RL}$-submodules of $\RL^2$.
\end{remark}

%For $\RL$ and $\RL^2$, the following  will be quoted.

\begin{remark} \label{rk R semisimple} \rm
If $\gcd(n,q)=1$, then the algebra $\RL$ is semisimple, cf. \cite{CDFL}; and
by Ring Theory (e.g., cf. \cite{J} or \cite[Remark 2.4]{FZ23}),
we have the following two.

(1) For any ideal ($\RL$-submodule) $C$ of ${\RL}$,  %(denote by $C\leq\RL$),
there is an idempotent $e_C$ of~${\RL}$ such that
$C=\RL e_C$ and ${\RL}=C\oplus D$ where $D={\RL}(1-e_C)$.
Note that $C=\RL e_C$ is an algebra with identity $e_C$ (but not a subalgebra of $\RL$
in general because $e_C\ne 1$ in general); in particular, $C^\times$ makes sense.
Moreover, if $I=\RL f$ with $f$ being an idempotent and
an ideal $C\subseteq I$, then $I=C\oplus C'$
with $C'=\RL(f-e_C)$.

(2) If $C$ and $C'$ are $\RL$-submodules of $\RL$; and
$\varphi: C\to C'$ is an $\RL$-module isomorphism,
then $C'=C$, and there is a $g\in C^{\times}$
%(recall that $C=\RL e$ is an ring with identity $e$)
such that
$\varphi(c)=cg$ for any $c\in C$.
\end{remark}

If $\lambda=1$, there is another identification.
By Eq.(\ref{eq R_lambda}), we denote
\begin{align} \label{eq R=R_1}
 \R1=\R1_1=F[X]/\langle X^n-1\rangle.
\end{align}
Let $G=\langle\, x\,|\,x^n=1\,\rangle=\{1,x,\dots,x^{n-1}\}$
be the cyclic group of order $n$.
Let $FG$ be the cyclic group algebra, i.e.,
$FG=\big\{\sum_{i=0}^{n-1}a_ix^i\,\big|\,a_i\in F\big\}$
 is an $F$-vector space with basis $G$ and equipped
with the multiplication induced by the multiplication of the group $G$ as follows:
$$
 \Big(\sum_{i=0}^{n-1}a_ix^i\Big)\Big(\sum_{j=0}^{n-1}b_jx^j\Big) =
 \sum_{k=0}^{n-1} %\Big(\sum_{i+j\equiv k\;({\rm mod}~n)} a_ib_j\Big)x^k
  \Big(\sum_{x^i x^j=x^k} a_ib_j\Big)x^k.
$$
There is a canonical algebra isomorphism:
\begin{align} \label{eq R to FG}
\R1=F[X]/\langle X^n-1\rangle
\;\mathop{\longrightarrow}^{\cong}\; FG, ~~~~
 \sum_{i=0}^{n-1} a_iX^i\;\longmapsto\;\sum_{i=0}^{n-1} a_ix^i.
\end{align}
Thus, $\R1$ is identified with the cyclic group algebra
$FG=\{\sum_{i=0}^{n-1}a_ix^i\,|\,a_i\in F\}$.
And by Remark~\ref{rk R ident. Fn}(1), the cyclic codes over $F$ of length $n$
are identified with the ideals of the cyclic group algebra $FG$.
Similarly, $2$-quasi-cyclic codes over $F$ of length $2n$
are identified with the $FG$-submodules of
the $FG$-module $(FG)^2=FG\times FG$.

With the identifications Eq.(\ref{eq R to FG}),
we have more algebraic preliminaries about $\R1$ to introduce.
Assume that $\gcd(n,q)=1$, then %${\gcd(n,q)=1}$,
$\R1$ is semisimple.
Let
\begin{align} \label{eq e_0 ...}
e_0=\frac{1}{n}\sum_{i=0}^{n-1}x^i,~ e_1,~ \dots,~ e_m
\end{align}
be all primitive idempotents of $\R1$. Correspondingly, %we have
the irreducible decomposition of $X^n-1$ in $F[X]$ is as follows
\begin{align} \label{eq X^n-1=...}
X^n-1=\varphi_0(X)\varphi_1(X)\dots \varphi_m(X), \quad
 \mbox{where } \varphi_0(X)=X-1,
\end{align}
such that
\begin{align} \label{eq R=...}
 \begin{array}{l}
 \R1 = \R1 e_0\oplus \R1 e_1 \oplus\dots\oplus \R1 e_m;\\[3pt]
\R1 e_i\cong F[X]/\langle\varphi_i(X)\rangle,\quad i=0,1,\dots,m.
 \end{array}
\end{align}
Since $\varphi_i(X)$ is irreducible over $F$,
each $\R1 e_i$ is an extension field over $F$ with identity $e_i$, and
$\dim_{F}\R1 e_i=\deg\varphi_i(X)$.
As $e_0=\frac{1}{n}\sum_{i=0}^{n-1}x^i$,
$xe_0=e_0$; so
\begin{align} \label{eq R e_0}
\R1 e_0 =\!
\big\{\alpha\!+\!\alpha x+\dots+\!\alpha x^{n-1}\,\big|\,\alpha\!\in\! F\big\}, ~
\dim_{F}\R1 e_0=1, ~ {\rm w}(\R1 e_0)=n.
\end{align}
For $1\leq i\leq m$, we denote
\begin{align} \label{eq mu(n)=...}
\mu(n):=\min_{1\le i\le m} \deg\varphi_i(X)
 =\min_{1\le i\le m} \dim_F\R1 e_i.
\end{align}

\begin{remark} \label{rk semisimple}\rm
In general, in this paper we consider $\RL$ for any $q$ and $n$
 (not restricted to the semisimple case)
unless the hypothesis ``$\gcd(n,q)=1$'' is explicitly assumed.
Once ``$\gcd(n,q)=1$'' is assumed, the above preliminaries on the
semisimple case can be quoted.
\end{remark}

\begin{remark} \label{rk 2-quasi cyclic} \rm
As mentioned in Introduction,
if $q\;{\not\equiv}\;3~({\rm mod}~4)$ then the
self-dual $2$-quasi-cyclic codes are asymptotically good.
To state it more precisely, we need the so-called {\em $q$-entropy function}:
\begin{align} \label{eq def h_q}
h_{q}(\delta)=
\delta \log_{q}(q\!-\!1)-\delta \log_{q}(\delta)-(1\!-\!\delta)\log_q(1\!-\!\delta),
 ~~~
 \delta\in[0,1\!-\!q^{-1}],
\end{align}
which value strictly increases from $0$ to $1$
while $\delta$ increases from $0$ to $1-q^{-1}$.
By \cite[Theorem IV.17]{LF22}, for any real number $\delta$ with
$0<\delta< 1-q^{-1}$ and $h_q(\delta)<\frac{1}{4}$,
there are self-dual $2$-quasi-cyclic codes $C_1,C_2, \dots$ over $F$
(hence ${\rm R}(C_i)=\frac{1}{2}$) such that:
(1) the relative minimum distance $\Delta(C_i)>\delta$\, for $i=1,2,\dots$; and
(2)~the code length $2n_i$ of $C_i$ satisfies that every $n_i$ is odd and coprime to $q$,
 and $\lim\limits_{i\to\infty} \frac{\log_q(n_i)}{\mu(n_i)}=0$
(in particular, $2n_i\to\infty$).
\end{remark}

\begin{definition}\label{def  balanced} \rm
Let $C\subseteq F^n=F^I$ with $I=\{0,1,\dots,n-1\}$ being the index set.
If there are positive integers $s,k,t$ and subsets (repetition is allowed)
$I_1,\dots, I_s$ of $I$ with $|I_j|=k$ for $j=1,\dots,s$
satisfying the following two conditions:
(1) for each $I_j=\{i_1,\cdots,i_k\}$ ($1\le j\le s$)
the projection $\rho_{I_{j}}$: $F^{I}~\to~ F^{I_{j}}$,
$(a_0,a_1,\dots,a_{n-1})\mapsto(a_{i_1},\dots,a_{i_k})$,
maps $C$ bijectively onto~$F^{I_j}$;
and (2) for any~$i$ ($0\le i< n$)
the number of the subsets $I_j$ which contains $i$ (i.e., $i\in I_j$) equals $t$;
then we say that $C$ is a {\em balanced code} over $F$ of length $n$,
and $I_1,\dots,I_s$ are called {\em information index sets} of the code~$C$.
\end{definition}

An important result (see \cite[Corollary 3.4]{FL15}) is that:
if $C\subseteq F^{n}$ is a balance code with cardinality $|C|=q^k$, then
$$
 |C^{\le\delta}|\le q^{kh_q(\delta)}, ~~~ \mbox{for}~~ 0\le\delta\le 1-q^{-1},
$$
where $h_q(\delta)$ is defined in Eq.\eqref{eq def h_q} and
\begin{align} \label{eq def C^<=}
 C^{\le\delta}=\{\,{\bf c}\,|\, {\bf c}\in C,\, {{\rm w}({\bf c})}\le\delta n\,\}.
\end{align}
Constacyclic codes are balanced, see~\cite[Lemma~II.8]{FL22}.
By \cite[Corollary 3.4 and Corollary 3.5]{FL15},
we can easily obtain the following lemma.

\begin{lemma}\label{lem balance}
If $A$ is an ideal of $\RL$, then the $\RL$-submodule $A\times A$ of $\RL^2$ is
a balanced code, hence for $0\le\delta\le 1-q^{-1}$,
$$ 
 |(A\times A)^{\le\delta}| \le q^{2\dim_F (A)\cdot h_q(\delta)}.
\eqno\qed
$$ 
\end{lemma}

\section{$2$-quasi constacyclic codes over finite fields}
\label{2-Q CC}

In this sections we are primarily concerned with the algebra properties of
 $2$-quasi $\lambda$-constacyclic codes over~$F$.
In the following, we always assume that
\begin{align} \label{eq t=...}
\lambda\in F^\times ~~~ \mbox{and}~~~ t={\rm ord}_{F^\times}(\lambda),
\end{align}
where ${\rm ord}_{F^\times}(\lambda)$ denotes
the order of $\lambda$ in the multiplication group $F^\times$.

As remarked in Remark~\ref{rk semisimple},
most of this section discusses $\RL$ for any $q$ and~$n$,
only Theorem~\ref{Goursat C} and its corollaries consider
the semisimple case (i.e., $\gcd(n,q)=1$).

\begin{remark} \label{rk projections} \rm
For $F^n\times F^n$, there are two projections $\rho_1$, $\rho_2$:
$F^n\times F^n \to F^n$ as follows
$$
 \rho_1({\bf a}, {\bf a}'\big)={\bf a},~~~
\rho_2({\bf a}, {\bf a}')={\bf a}', ~~~~~~
\forall~({\bf a},\, {\bf a}')\in F^n\times F^n.
$$
By the linear isomorphism Eq.\eqref{eq RR to FF},
the projections $\rho_1$ and $\rho_2$
are also defined on $\RL^2=\RL\times\RL$:
for $\big(a(X),\,a'(X)\big)\in\RL^2$,
\begin{align*}
 \rho_1\big(a(X),\,a'(X)\big)=a(X), ~~~~~
 \rho_2\big(a(X),\,a'(X)\big)=a'(X).
\end{align*}
For any $\RL$-submodule $C$ of ${\RL}\times{\RL}$,
restricting $\rho_1$ to $C$, we have an ${\RL}$-homomorphism
$\rho_1|_{C}: C\to\rho_1(C)$.
Observe that the kernel of the restricted homomorphism is
\begin{align} \label{kernel of rho}
 {\rm Ker}(\rho_1|_{C})=C\cap(0\times{\RL})
=\big\{\big(0,c'(X)\big)\in C\,\big|\,c'(X)\in{{\RL}}\big\}.
\end{align}
\end{remark}

It is known that for $\lambda_1\ne \lambda_2\in F^\times$,
if a linear code $C\subseteq F^n$ is both
$\lambda_1$-constacyclic and $\lambda_2$-constacyclic,
then either $C=\{{\bf 0}\}$ or $C=F^n$ (cf. {\cite{D12}}).
Extending the result to $F^n\times F^n$, we get the following lemma.
% (thanks are given to Dr. Dinh H.Q. who told the first author the result).

\begin{lemma} \label{lambda_1 lambda_2}
%With the notation as above.
Let $\lambda_1,\lambda_2\in F^\times$ such that $\lambda_1\ne\lambda_2$.
If a subspace $C$ of  $F^n\times F^n$ is
both a $2$-quasi $\lambda_1$-constacyclic code and
 a $2$-quasi $\lambda_2$-constacyclic code, then either $\rho_1(C)=\{{\bf 0}\}$ or $\rho_1(C)=F^n$;
 and it is the same for $\rho_2(C)$.
\end{lemma}
\begin{proof}
Suppose $\rho_1(C)\ne \{{\bf 0}\}$. There is a codeword
$$
 ({\bf c},\,{\bf c}')=(c_0,\dots, c_{n-1}, \; c'_0,\dots, c'_{n-1}) \,\in\, C
$$
with some $c_j\ne 0$, $0\le j<n$; so we can assume that $c_{n-1}\ne 0$.
By the double $\lambda$-constacyclic permutation in %it $n-i$ times
Eq.\eqref{int def 2 lambda-cyclic},
then $C$ contains the following word
$$
\frac{1}{(\lambda_1-\lambda_2)c_{n-1}}
\Big(\big(\mathbf{c} P_{\!\lambda_1},\,\mathbf{c}'P_{\!\lambda_1}\big)
-\big(\mathbf{c} P_{\!\lambda_2},\,\mathbf{c}'P_{\!\lambda_2}\big)\Big)
=\big( 1,0,\dots,0,\,d'_0,\dots,d'_{n-1}\big),
$$
for some  $d'_{i}\in{F}$, $i=0,1,\dots,n-1$.
Using the double $\lambda$-constacyclic permutation $i$ times, we see that
$C$ contains such words
$$
 \big( \mathop{0,\dots,0,\, 1,\,0,\dots, 0,}
 \limits_{\mbox{\scriptsize index\,$i$}}\,b'_0,\dots,b'_{n-1}\big),
 \qquad i=0,1,\dots,n-1,
$$
for some  $b'_{i}\in{F}$, $i=0,1,\dots,n-1$.
It follows that $\rho_1(C)$ contains a basis of~$F^n$, hence $\rho_1(C)=F^n$.

By the same argument, either $\rho_2(C)=\{{\bf 0}\}$ or $\rho_2(C)=F^n$.
\end{proof}

For any matrix $A=\big(a_{ij}\big)_{m\times n}$ with $a_{ij}\in F$,
we denote $A^{(p^h)}=\big(a_{ij}^{p^h}\big)_{m\times n}$
(cf. Example~\ref{exm}(2)).
By $A^T$ we denote the transpose of~$A$.
Let
\begin{align} \label{eq def *h}
A^{*h}=(A^{(p^h)})^T =\big(a_{ji}^{p^h}\big)_{n\times m}
 =\begin{pmatrix} a_{11}^{p^h} & \cdots & a_{m1}^{p^h} \\
   \dots & \dots & \dots \\ a_{1n}^{p^h} & \cdots & a_{mn}^{p^h} \end{pmatrix}
\end{align}
be the transpose of the matrix~$A^{(p^h)}$.
We call $A^{*h}$ the {\em Galois $p^h$-transpose} of $A$.
If $h=0$, $A^{*0}=A^T$ is just the transpose matrix of~$A$.
If $\ell$ is even and $h=\frac{\ell}{2}$, $A^{*\frac{\ell}{2}}$ is the Hermitian transpose
of $A$.

For $B=(b_{ij})_{m\times n}$, $C=(c_{ij})_{n\times k}$ and $\alpha\in F$,
it is easy to check that
\begin{align} \label{eq *h satisfy}
(AC)^{*h}=C^{*h}A^{*h}, ~~  (A+B)^{*h}=A^{*h}+B^{*h}, ~~
 (\alpha A)^{*h}=\alpha^{p^h} A^{*h}.
\end{align}
Hence, if $m=n$, the operator ``$*h$'' is a $p^h$-anti-automorphism
of the matrix algebra ${\rm M}_{n}(F)$
(compare it with Eq\eqref{eq alg iso} and Example~\ref{exm}(2)).

With the identification Eq.\eqref{eq R to F} and the operator ``$*h$'',
 we can compute the Galois $p^h$-inner product on ${\RL}$
(see Eq.\eqref{int Galois inner})
in a matrix version:
\begin{align}\label{eq G1 inner prod.}
 \big\langle a,\,b\big\rangle_h =\mathbf{a}\cdot \mathbf{b}^{*h},
  \qquad \forall~ a,b\in{\RL}.
\end{align}
And for any $\lambda$-constacyclic code $C$  (i.e., any idea of $\RL$),
the Galois $p^h$-dual code of $C$ is as follows:
$$C^{\bot h}=\big\{\,a\in{{\RL}}~\big|~ \big\langle c,\,a\big\rangle_h=0,
~ \forall\, c\in{C}\,\big\}.
$$
Similarly, with the identification Eq.\eqref{eq RR to FF}
the Galois $p^h$-inner product on $\RL^{2}$ is computed in a matrix version:
\begin{align} \label{inner as matrix product}
 \big\langle (a, a'),\,(b, b')\big\rangle_h
=\big(\mathbf{a},\mathbf{a}'\big)\cdot\big(\mathbf{b},\mathbf{b}'\big)^{*h},
\qquad \forall ~\big(a,a'\big),\,\big(b,b'\big)\in\RL^2.
\end{align}
And for any $2$-quasi $\lambda$-constacyclic code $C$ ($\RL$-submodule of $\RL^{2}$),
the Galois $p^h$-dual code of $C$ is
$$C^{\bot h}=\big\{\,(a,a')\in {\RL^{2}}~\big|~
 \big\langle (c, c'),\,(a, a')\big\rangle_h=0,~ \forall\,(c,c')\in{C}\,\big\}.$$
If $C=C^{\bot h}$, then we say that $C$ is {\em Galois $p^h$-self-dual},
or {\em Galois self-dual}.

\begin{lemma} \label{G C^bot lambda^-1}
If $C\subseteq F^n\times F^n$ is a $2$-quasi $\lambda$-constacyclic code,
then $C^{\bot h}$ is a $2$-quasi $\lambda^{\!-\!p^{\ell-h}}\!$-constacyclic code.
In particular, if $h=0$,
$C^{\bot 0}=C^\bot$ is a $2$-quasi $\lambda^{\!-\!1}$-constacyclic code.
\end{lemma}

\begin{proof}
Assume that $\big(\mathbf{a},\, \mathbf{a}'\big)\in C^{\bot h}$.
By Lemma~\ref{lambda-CC code}(2), it is enough to prove that
for any $\big(\mathbf{c}, \mathbf{c}'\big)\in C$
we have
$$
\left\langle\big(\mathbf{c},\, \mathbf{c}\big),\,\big(\mathbf{a},\, \mathbf{a}'\big)
\begin{pmatrix}P_{\lambda^{\!-\!p^{\ell-h}}}\!\\ &\! P_{\lambda^{\!-\!p^{\ell-h}}}
  \end{pmatrix}\right\rangle_{\!h}=0.
$$
Applying Eq.\eqref{eq *h satisfy} and Eq.\eqref{inner as matrix product} yields
\begin{align*}
&\left\langle \big(\mathbf{c}, \mathbf{c}'\big),\,\big(\mathbf{a},\, \mathbf{a}'\big)
\begin{pmatrix}P_{{\lambda^{\!-\!p^{\ell-h}}}}\!\\ &\!
 P_{{\lambda^{\!-\!p^{\ell-h}}}} \end{pmatrix}\right\rangle_{\!h}
\\
&=\big(\mathbf{c}, \mathbf{c}'\big)\cdot\left(
 \big(\mathbf{a},\, \mathbf{a}'\big)
 \begin{pmatrix}P_{{\lambda^{\!-\!p^{\ell-h}}}}\!\\
 &\! P_{{\lambda^{\!-\!p^{\ell-h}}}}\end{pmatrix}\right)^{*h} \\
&=\big(\mathbf{c}, \mathbf{c}'\big)\begin{pmatrix}P_{\lambda^{\!-\!p^{\ell-h}}}\!\\
 &\! P_{\lambda^{\!-\!p^{\ell-h}}}\end{pmatrix}^{*h}\big(\mathbf{a},\,
 \mathbf{a}'\big)^{*h}\\
&=\big(\mathbf{c},
\mathbf{c}'\big)\begin{pmatrix}\big(P_{\lambda^{\!-\!p^{\ell-h}}}\big)^{*h}\!\\
 &\! \big(P_{\lambda^{\!-\!p^{\ell-h}}}\big)^{*h}\end{pmatrix}\big(\mathbf{a},\,
 \mathbf{a}'\big)^{*h}.
\end{align*}
Because $(\lambda^{\!-\!p^{\ell-h}})^{p^h}=(\lambda^{p^{\ell}})^{-1}=\lambda^{-1}$,
it is easy to check that
$$
P_{\!\lambda}\big(P_{\lambda^{\!-\!p^{\ell-h}}}\big)^{*h}
=\begin{pmatrix} & E_{n-1}\\ \lambda E_1\end{pmatrix}
\begin{pmatrix} &\! (\lambda^{\!-\!p^{\ell-h}})^{p^h}E_{1}\\
  E_{n-1}\!\end{pmatrix} = E_{n}.
$$
Since $P_{\!\lambda}^n=\lambda E_n$ and ${\rm ord}_{F^{\times}}(\lambda)=t$
(see Eq\eqref{eq t=...}),
we have $P_{\!\lambda}^{nt}=E_n$,
hence $P_{\!\lambda}^{-1}=P_{\!\lambda}^{nt-1}$.
So
\begin{align} \label{G P_lambda^-1}
\big(P_{\lambda^{\!-\!p^{\ell-h}}}\big)^{*h}
=P_{\!\lambda}^{-1}=P_{\!\lambda}^{nt-1}.
\end{align}
By Lemma \ref{lambda-CC code}(2),
$\big(\mathbf{c}, \mathbf{c}'\big)
\begin{pmatrix}P_{\!\lambda}\!\\ &\! P_{\!\lambda}
  \end{pmatrix}^{nt-1}\in C$. We get that
\begin{align*}
\bigg\langle\big(\mathbf{c}, \mathbf{c}'\big),\,\big(\mathbf{a}, \mathbf{a}'\big)
\begin{pmatrix}P_{\!\lambda^{\!-\!p^{\ell-h}}}\!\!\\
 &\!\! P_{\!\lambda^{\!-\!p^{\ell-h}}}
  \end{pmatrix}\bigg\rangle_{\!h}
=\bigg\langle\big(\mathbf{c}, \mathbf{c}'\big)
\begin{pmatrix}P_{\!\lambda}\!\\ &\! P_{\!\lambda}
  \end{pmatrix}^{nt-1}\!,\,
\big(\mathbf{a}, \mathbf{a}'\big)\bigg\rangle_{\!h}=0.
\end{align*}
We are done.
\end{proof}

\begin{remark} \label{lambda^1+p^h} \rm
 We note that
 $\lambda=\lambda^{-p^{\ell-h}}$ if and only if $\lambda^{1+p^{h}}=1$, since
$$
\lambda=\lambda^{-p^{\ell-h}} \;\iff\;
\lambda^{p^h}=(\lambda^{-p^{\ell-h}})^{p^h}=\lambda^{-p^{\ell}}=\lambda^{-1}
\;\iff\; \lambda^{1+p^{h}}=1.
$$
If $\lambda=\lambda^{-p^{\ell-h}}$, then
for any $2$-quasi $\lambda$-constacyclic code $C$,
its Galois $p^h$-dual code $C^{\bot h}$ is
still a $2$-quasi $\lambda$-constacyclic code.
Otherwise (i.e., $\lambda\ne\lambda^{-p^{\ell-h}})$,
Lemma \ref{lambda_1 lambda_2} implies that,
for many $2$-quasi $\lambda$-constacyclic codes,
their Galois $p^h$-dual codes are no longer $2$-quasi $\lambda$-constacyclic codes.
\end{remark}

%We introduce a kind of $2$-quasi constacyclic codes.
\begin{remark} \label{rk C_a} \rm
For  $a(X)=\sum_{i=0}^{n-1}a_iX^i\in {\RL}$,
by $C_{a}$ we denote the ideal
of ${\RL}$ generated by  $a(X)$.
Similarly, for $(a(X),a'(X))\in \RL^2$,
by $C_{a,a'}$ we denote the ${\RL}$-submodule of $\RL^2$
generated by $\big(a(X),\,a'(X)\big)$.
Note that %in the semisimple case
 any ideal of $\RL$ is generated by one element
 (cf. Remark~\ref{rk R semisimple}(1) for semisimple case,
and cf. \cite[Lemma 4.3]{F21} for general case).
However, some $\RL$-submodules of $\RL^2$
can not be generated by one element. For example, as an $\RL$-submodule
$\RL^2$ can not be generated by one element
(because any $\RL$-module generated by one element
is a quotient of the regular module).
\end{remark}

\begin{definition} \label{def consta-circulant} \rm
For any $a(X)=a_0+a_1X+\cdots+a_{n-1}X^{n-1}\in{F[X]}$
with $\deg(a(X))<n$,
we have an $n\times n$ matrix
\begin{align} \label{lambda-circulant matrix}
a(P_{\!\lambda})=a_0E+a_1P_{\!\lambda}+\cdots+a_{n-1}P_{\!\lambda}^{n-1}
=\begin{pmatrix}
 a_0 & a_1 & \cdots & a_{n-1} \\
 \lambda a_{n-1} & a_0 &\cdots & a_{n-2} \\
 \cdots & \cdots &\cdots & \cdots \\
 \lambda a_{1} & \lambda a_2 &\cdots & a_{0}
\end{pmatrix},
\end{align}
whose first row is the vector $\mathbf{a}=(a_0,a_1,\dots,a_{n-1})$,
and each next row is obtained
by $\lambda$-constacyclically permuting the present row (cf. Eq.\eqref{eq a dot P}).
We call $a(P_{\!\lambda})$
the {\em $\lambda$-consta circulant matrix}
associated with the polynomial $a(X)$.
\end{definition}

\begin{lemma} \label{lem C_a,a'}
Let $a(X),a'(X)\in\RL$. Then we have:

{\bf(1)}~ $C_{a}$ is linearly generated by the rows of the $n\times n$
matrix $a(P_{\!\lambda})$.

{\bf(2)}~ $C_{a,a'}$ is linearly generated by the rows of the $n\times 2n$
matrix $\big(a(P_{\!\lambda}), a'(P_{\!\lambda})\big)$.
\end{lemma}

\begin{proof}
(1). Let $C_a$ be the ideal of $\RL$ generated by $a(X)$, i.e.,
$$C_a=\Big\{f(X)a(X)~\Big|~f(X)=\sum_{i=0}^{n-1}f_iX^i\in\RL, f_i\in F\Big\}.$$
Obviously,
$f(X)a(X)=f_0\cdot a(X)+f_1\cdot Xa(X)+\dots+f_{n-1}\cdot X^{n-1}a(X)$.
So~$C_a$ is the subspace of $\RL$ linearly generated by
$a(X), Xa(X), \dots, X^{n-1}a(X)$.
Eq.\eqref{eq R cong F} and Eq.\eqref{lambda-circulant matrix} imply that
 $a(X)$ is identified with the row vector ${\bf a}$
which is just the first row of the matrix $a(P_\lambda)$;
for $i=1, \dots, n-1$,
 $X^{i}a(X)$ is identified with the row vector ${\bf a}\cdot P_\lambda^{\,i}$,
which is just the $i$'th row of the matrix $a(P_\lambda)$.
Therefore, $C_a$ is linearly generated by the rows of the matrix $a(P_{\lambda})$.

Obviously, (2) is proved in a similar way.
\end{proof}

\begin{corollary}\label{Cor.C1,a}
For $a(X)\in{\RL}$, the $2$-quasi $\lambda$-constacyclic code
$C_{1,a}$ has a generating matrix $\big(E_n,\,a(P_{\!\lambda})\big)$.
\end{corollary}
\begin{proof}
By Lemma~\ref{lem C_a,a'}, the rows of the matrix
$\big(E_n,\,a(P_{\!\lambda})\big)$ linearly generate the  code $C_{1,a}$.
And the rows of the matrix $\big(E_n,\,a(P_{\!\lambda})\big)$
are linearly independent.
\end{proof}

In the rest of this section, we turn to the semisimple case, i.e., $\gcd(n,q)=1$.
Extending \cite[Theorem 3.2]{FZ23} which characterized the algebraic structure
of $2$-quasi-cyclic codes in the semisimple case, %in that case
we characterize the algebraic structure of $2$-quasi constacyclic
codes as follows.

\begin{theorem}\label{Goursat C}
Assume that $\gcd(n,q)=1$.
If $C$ is an $\RL$-submodule of $\RL^2$,
then there are ideals $C_1,C_2,C_{12}$ of ${\RL}$
satisfying that $C_1\cap C_{12}=C_2\cap C_{12}=0$
and an element $g\in C_{12}^\times$ such that
\begin{equation}\label{structure C}
 C = (C_1\times C_2)\oplus\widehat C_{12},~~ \mbox{where}~~
\widehat C_{12}=
 \big\{\big(c_{12},c_{12}g\big)\,\big|\,c_{12}\in C_{12}\big\}\cong C_{12}.
\end{equation}
Conversely, if there are ideals $C_1, C_2, C_{12}$ of ${\RL}$
with $C_1\cap C_{12}=C_2\cap C_{12}=0$ and an element
$g\in C_{12}^\times$, then $C$ in Eq.\eqref{structure C}
is an $\RL$-submodule of $\RL^{2}$.
\end{theorem}

\begin{proof}
The ``conversely'' part is obviously true because
both $C_1\times C_2$ and $\widehat C_{12}$
are $\RL$-submodules of $\RL^2$, and
$(C_1\times C_2)\cap \widehat C_{12}=0$.

Assume that $C$ is an $\RL$-submodule of $\RL^2$,
and  $\rho_1$, $\rho_2$
are defined in Remark~\ref{rk projections}.
We consult the module version of {Goursat Lemma}
(see\cite[Remark~3.1]{FZ23}). Take
\begin{align}\label{}
\begin{array}{l}
\tilde{C_1}=\rho_1(C)=\big\{a_1\in{{\RL}}~|~(a_1, a_2)\in{C} ~
 \mbox{for some} ~a_2\in{{\RL}}\big\},\\[2pt]
{C_1}=\rho_1(C\cap({\RL} \times 0))=\big\{a_1\in{{\RL}}~|~(a_1, 0)\in{C}\big\},\\[2pt]
\tilde{C_2}=\rho_2(C)=\big\{a_2\in{{\RL}}~|~(a_1, a_2)\in{C} ~\mbox{for some} ~a_1\in{{\RL}}\big\},\\[2pt]
{C_2}=\rho_2(C\cap(0 \times {\RL}))=\big\{a_2\in{{\RL}}~|~(0, a_2)\in{C}\big\}.
\end{array}
\end{align}
Then $C_i\subseteq \tilde{C_i}$ are ideals of ${\RL}$, $i= 1, 2$.
For any $c_1+C_1\in{\tilde{C_1}/C_1}$, there is a unique $c_2 +C_2 \in{\tilde{C_2}/C_2}$
such that $(c_1, c_2)\in{C}$, hence we have the map
\begin{equation}\label{varphi}
 \vph: \tilde C_1/C_1\to \tilde C_2/C_2, ~
 c_1+C_1\mapsto c_2+C_2,
\end{equation}
which is an ${\RL}$-isomorphism,  and
\begin{equation}\label{structure}
 C=\big\{ (c_1,c_2)\;\big|\; c_i\in\tilde C_i~{\rm for} ~i=1,2; ~\vph(c_1+C_1)=c_2+C_2\big\}.
\end{equation}
Since ${\RL}$ is semisimple and
${C}_1\subseteq \tilde C_1$ are ideals of ${\RL}$,
there is an ideal $C_{12}$ of $\RL$ such that
$\tilde C_1=C_1\oplus C_{12}$; see Remark~\ref{rk R semisimple}(1).
Similarly, we have an ideal $C_{12}'$ of ${\RL}$ such that
$\tilde C_2=C_2\oplus C_{12}'$.
Then $\tilde C_1/C_{1}\cong C_{12}$ and $\tilde C_2/C_{2}\cong C'_{12}$.
The ${\RL}$-isomorphism $\vph$ in Eq.\eqref{varphi} induces an
${\RL}$-isomorphism $\vph': C_{12}\to C_{12}'$ such that
$\vph\big(c+C_1\big)=\vph'(c)+C_2$ for all $c\in C_{12}$.
Thus the image $C_{12}'=C_{12}$, and there is a $g\in C_{12}^\times$
such that $\vph'(c)=cg$ for all $c\in C_{12}$; cf. Remark~\ref{rk R semisimple}(2).
In conclusion,
we have an ideal $C_{12}$ of $\RL$ such that
$ %\begin{equation*}
 \tilde C_1 =C_1\oplus C_{12},~
 \tilde C_2 =C_2\oplus C_{12};
$ %\end{equation*}
and a $g\in C_{12}^\times$ such that
\begin{equation*} %\label{more structure C}
 C=\big\{\,(c_1+c_{12},\;c_2+c_{12}g)\;\big|\;
 c_1\in C_1, \, c_2 \in C_2,\,  c_{12}\in C_{12}\,\big\}.
%\!=\!(C_1\times C_2)\oplus \widehat C_{12}.
\end{equation*}
Obviously, $(c_1\!+\!c_{12},\;c_2\!+\!c_{12}g)=(c_1,c_2)+(c_{12}, c_{12}g)$.
Thus Eq.\eqref{structure C} holds.
\end{proof}

\begin{lemma}\label{lem Ca cap Cb}
Assume that $\gcd(n,q)=1$. Then for $a, b\in{\RL}$,
$C_a\cap C_b=0$ if and only if $ab=0$.
\end{lemma}
\begin{proof}
If $C_a\cap C_b=0$ then $ab\in C_a\cap C_b =0$.
Conversely, assume that $ab=0$.
Note that $C_b=\RL\, e_b$ for an idempotent~$e_b$
which is the identity of the ring~$C_b$
(cf. Remark~\ref{rk R semisimple}(1)).
There is an element $f\in\RL$ such that $e_b=fb$.
For any element $c \in{C_a\cap C_b}$, we can write $c=ra=sb$ with $r,\,s\in{\RL}$.
Then $c=ra=rae_b=rafb=0$. Thus  $C_a\cap C_b=0$.
\end{proof}

\begin{corollary}\label{lem structure C again}
Keep the notation in Theorem~\ref{Goursat C}
(in particular, ${\gcd(n,q)=1}$).
Any $\RL$-submodule $C$ of $\RL^2$ can be
 written as
 \begin{align} \label{eq structure C again}
C= (C_a\times C_{a'})\oplus C_{b,bg}
=C_{a,0}\oplus C_{0,a'}\oplus C_{b,bg},
\end{align}
 where $a,a',b\in\RL$ and $g\in C_{b}^\times$ satisfy that
$ab=a'b=0$.
\end{corollary}
\begin{proof}
 Take
$C_1=C_a$, $C_2=C_{a'}$ and
$C_{12}=C_b$, hence $\widehat C_{12}=C_{b,bg}$.
By the above Lemma, we get Eq.(\ref{eq structure C again}) immediately.
\end{proof}

\section{Galois self-dual $2$-quasi constacyclic codes}
\label{Section Galois duality}

In this section
we investigate the Galois $p^h$-self-dual
$2$-quasi $\lambda$-constacyclic codes over $F$.
By Lemma~\ref{G C^bot lambda^-1} and Remark~\ref{lambda^1+p^h},
any Galois $p^h$-self-dual $2$-quasi $\lambda$-constacyclic code
(i.e., $C=C^{\bot h}$) is also $\lambda^{-p^{\ell-h}}$-constacyclic.
We study them in two cases:
$\lambda \neq \lambda^{-p^{\ell-h}}$ (i.e., $\lambda^{1+p^h} \neq 1$),
or $\lambda = \lambda^{-p^{\ell-h}}$ (i.e., $\lambda^{1+p^h} = 1$).

Still, this section discusses $\RL$ for any $q$ and~$n$
except for Theorem~\ref{th lambda eq iff} which considers
the semisimple case (i.e., $\gcd(n,q)=1$).

\subsection{The case that $\lambda^{1+p^h}\neq 1$}
%$\lambda\neq\lambda^{-p^{\ell-h}}$}
Our concern in this subsection is
the Galois $p^h$-self-dual 2-quasi $\lambda-$constacyclic codes over $F$
under the assumption that $\lambda^{1+p^{h}}\ne 1$.
%($\lambda\neq\lambda^{-p^{\ell-h}}$ equivalently).

\begin{theorem}\label{th lambda neq iff}
Assume that $\lambda^{1+p^{h}}\ne 1$
 (i.e., $\lambda\ne\lambda^{-p^{\ell-h}}$).
The following three are equivalent to each other:

{\bf(1)} $C$ is a Galois $p^h$-self-dual $2$-quasi $\lambda$-constacyclic code
over $F$ of length~$2n$. %, i.e., $C=C^{\bot h}$.

{\bf(2)}
$C=C_{1,\alpha}$ is a ${\RL}$-submodule of $\RL^2$
generated by $(1,\alpha)$, where
$\alpha\in F^{\times}$ with $\alpha^{1+p^h}=-1$, and
$1,\alpha$ are viewed as the constant polynomials of ${\RL}$.

{\bf(3)}
 $C$ is an $F$-linear code of length $2n$ with a generating matrix
$\big( E_n,\alpha E_n\big)$, where $\alpha\in F^\times$
with $\alpha^{1+p^h}=-1$.
\end{theorem}
\begin{proof}
Observe that by Corollary~\ref{Cor.C1,a},
the statements (2) and (3) are equivalent. Suppose that (3) holds.
%$C$ is a linear code of length $2n$ with a generating matrix
%$\big( E_n,\,\alpha E_n\big)$, where $\alpha\in F^\times$ with $\alpha^{1+p^h}=-1$.
Then
$
 \big( E_n,\,\alpha E_n\big)\big( E_n,\,\alpha E_n\big)^{*h}
 =(1+\alpha^{p^h+1})E_n=0,
$
which implies $\langle C,C\rangle_h=0$, cf. Eq.\eqref{inner as matrix product}.
Since the rank of the generating matrix $\big( E_n,\,\alpha E_n\big)$ is $n$,
we have $\dim_{F}C=n$, and so the statement (1) follows.
Therefore, it suffices to show that (1) implies (3).

Assume that (1) holds, i.e., $C=C^{\bot h}$.
By Lemma \ref{G C^bot lambda^-1},
$C$ is both a $2$-quasi $\lambda$-constacyclic code
and a $2$-quasi $\lambda^{\!-\!p^{\ell-h}}$-constacyclic code.
Since $\lambda\ne\lambda^{-p^{\ell-h}}$,
we deduce from Lemma~\ref{lambda_1 lambda_2} that
either $\rho_1(C)=0$ or $\rho_1(C)={\RL}$.
Suppose that $\rho_1(C)=0$.
Since $\dim_{F} C=n$,  by %Theorem~\ref{Goursat C} or
Eq.\eqref{kernel of rho},
we have that $C=0\times{\RL}$,
which is impossible because $0\times{\RL}$
is not Galois $p^{h}$-self-dual.
So, it must be the case that $\rho_1(C)={\RL}$.
By the same argument, we have $\rho_2(C)={\RL}$.

As $\rho_1(C)={\RL}$, we can take $(1,a)\in C$
with $a=a(X){=\sum_{i=0}^{n-1}a_iX^i\in{{\RL}}}$.
Let $C_{1,a}$ be the ${\RL}-$submodule
of $\RL^2$ generated by $(1,a)$.
Obviously, $C_{1,a}\subseteq C$.
By Corollary~\ref{Cor.C1,a},  $C_{1,a}$ has the generating matrix
\begin{align} \label{G lambda a matrix}
\big(E_{n},\;a(P_{\lambda})\big)=
\begin{pmatrix}
1 &&& && a_0 & a_1 & \cdots & a_{n-1} \\
 & 1 && && \lambda a_{n-1} & a_0 &\cdots & a_{n-2} \\
&  &\ddots & && \cdots & \cdots &\cdots & \cdots \\
&  && 1 && \lambda a_{1} & \lambda a_2 &\cdots & a_{0}
\end{pmatrix}_{n\times 2n}.
\end{align}
The rank of the matrix is $n$, hence $\dim_{F} C_{1,a}=n$. In a word,
$C=C_{1,a}$ is linearly generated by the $n$ rows of
the matrix in Eq.\eqref{G lambda a matrix}.

Because $C$ is also a $2$-quasi $\lambda^{\!-\!p^{\ell-h}}\!$-constacyclic code,
by the same way we can also get that
$C$ is linearly generated by the $n$ rows of the matrix
$\big(E_{n}, a(P_{\!\lambda^{-p^{\ell-h}}})\big)$.
Thus any row of $\big(E_{n},a(P_{\!\lambda^{-p^{\ell-h}}})\big)$ is a linear combination of
the rows of the matrix in Eq.\eqref{G lambda a matrix}.
So there is an $n\times n$ matrix $A$ such that
\begin{align*}
\big(E_{n},\;a(P_{\!\lambda^{-p^{\ell-h}}})\big)=A\Big(E_{n},\;a(P_{\lambda})\Big)
=\Big(AE_{n},\;A\cdot a(P_{\lambda})\Big).
\end{align*}
It follows that $A=E_{n}$, and so $a(P_{\!\lambda^{\!-\!p^{\ell-h}}})
=A\cdot a(P_{\!\lambda})=a(P_{\!\lambda})$. The latter equality is as follows:
\begin{align*}
\begin{pmatrix}
 a_0 & a_1 & \cdots & a_{n-1} \\
 \lambda^{\!-\!p^{\ell-h}} a_{n-1} & a_0 &\cdots & a_{n-2} \\
 \cdots & \cdots &\cdots & \cdots \\
\lambda^{\!-\!p^{\ell-h}} a_{1} & \lambda^{\!-\!p^{\ell-h}} a_2 &\cdots & a_{0}
\end{pmatrix}
=
\begin{pmatrix}
 a_0 & a_1 & \cdots & a_{n-1} \\
 \lambda a_{n-1} & a_0 &\cdots & a_{n-2} \\
 \cdots & \cdots &\cdots & \cdots \\
 \lambda a_{1} & \lambda a_2 &\cdots & a_{0}
\end{pmatrix}.
\end{align*}
Since $\lambda^{-p^{\ell-h}}\ne\lambda$,
it follows that $a_i=0$ for $i=1,\dots,n-1$;
hence the polynomial $a(X)$ is a constant polynomial: $a(X)=\alpha$
for some $\alpha\in{F}$,
and $C$ has a generating matrix $\big(E_{n},\,\alpha E_{n}\big)$.
Because $C$ is Galois $p^h$-self-dual,
\begin{align*}
0=\big(E_{n},\,\alpha E_{n}\big)\cdot\big(E_{n},\,\alpha E_{n}\big)^{*h}
=E_{n}+\alpha^{1+\alpha^{p^h}}E_{n}=(1+\alpha^{1+p^h})E_{n},
\end{align*}
hence $1+\alpha^{1+p^h}=0$. In conclusion, (3) holds.
\end{proof}

\begin{corollary}
Assume that $\lambda^{1+p^h}\ne 1$.
%By Theorem~\ref{th lambda neq iff} we have that:
The Galois $p^h$-self-dual $2$-quasi $\lambda$-constacyclic codes
of length $2n$ exist
if and only if the polynomial $X^{1+p^h}+1$ has roots in~$F$;
and in that case,
the ${\RL}$-submodules $C_{1,\alpha}$ of $\RL^2$ with
$\alpha\in F$ being a root of $X^{1+p^h}+1$ are all the Galois $p^h$-self-dual
$2$-quasi $\lambda$-constacyclic codes of length $2n$.
\qed
\end{corollary}

Note that in the Euclidean case ``$h=0$'',
$\lambda^{1+p^h}\neq 1$ if and only if $\lambda\ne\pm 1$.

\begin{corollary}
Assume that $\lambda\ne\pm 1$.
The self-dual $2$-quasi $\lambda$-constacyclic codes over $F$ of length $2n$
exist if and only if
$q\,{\not\equiv}\,3~({\rm mod}~4)$;
and in that case,
the ${\RL}$-submodules $C_{1,\alpha}$ of $\RL^2$ with $\alpha\in F$ satisfying
$\alpha^2=-1$ are the all self-dual $2$-quasi
$\lambda$-constacyclic codes over $F$ of length $2n$ .
\end{corollary}

\begin{proof}
The polynomial $X^2+1$ has roots in $F$ if and only if
$q$ is even or $q$ is odd and $4\,|\,(q-1)$,
if and only if $q\,{\not\equiv}\,3~({\rm mod}~4)$.
\end{proof}

As a comparison, the Hermitian self-dual ones always exist.

\begin{corollary}
Assume that $\ell$ is even and $\lambda^{1+p^{\ell/2}}\neq 1$.
Then the Hermitian self-dual $2$-quasi $\lambda$-constacyclic codes over $F$
of length $2n$ always exist;
and in that case, the ${\RL}$-submodules $C_{1,\alpha}$ of $\RL^2$ with
$\alpha\in F$ satisfying $\alpha^{1+p^{\ell/2}}=-1$ are the all Hermitian self-dual
$2$-quasi $\lambda$-constacyclic codes over $F$ of length $2n$.
\end{corollary}

\begin{proof}
If $q=p^\ell$ is even,
then the polynomial $X^{1+p^{\ell/2}}+1=X^{1+p^{\ell/2}}-1$
always has roots in $F$.
Assume that $p$ is odd. The order of the multiplication group
$$|F^\times|=p^\ell-1=(p^{\ell/2}+1)(p^{\ell/2}-1), $$
and $2\,\big|\,(p^{\ell/2}-1)$. So, there is a subgroup $H$
of the multiplication group $F^\times$ with order $|H|=2(p^{\ell/2}+1)$.
Then any generator of the group $H$ is a root of
the polynomial $X^{1+p^{\ell/2}}+1$.
\end{proof}

As a consequence, we get the following.

\begin{theorem} \label{thm case not good}
Assume that $\lambda^{1+p^h}\ne 1$ (i.e., $\lambda\ne\lambda^{-p^{\ell-h}}$).
Then the Galois $p^h$-self-dual $2$-quasi $\lambda$-constacyclic codes over $F$
are asymptotically bad.
\end{theorem}

\begin{proof}
If the polynomial $X^{1+p^h}+1$ has no root in $F$, then
Galois $p^h$-self-dual $2$-quasi $\lambda$-constacyclic codes over $F$ do not
exist, hence Galois $p^h$-self-dual $2$-quasi $\lambda$-constacyclic codes over $F$
are asymptotically bad.

Otherwise, any Galois $p^h$-self-dual $2$-quasi $\lambda$-constacyclic code
$C_{1,\alpha}$ has minimum weight ${\rm w}(C_{1,\alpha})=2$,
because any row of the generating matrix $\big(E_n, \alpha E_n\big)$
has weight $2$. The relative minimum distance
$\Delta(C_{1,\alpha})=\frac{2}{2n}\to 0$ while ${n\to\infty}$.
So Galois $p^h$-self-dual $2$-quasi $\lambda$-constacyclic codes over $F$
are asymptotically bad.
\end{proof}

\subsection{The case that $\lambda^{1+p^h}=1$} \label{sub lambda=}

We start with a general result about the $\lambda$-consta circulant matrices
(cf. Definition~\ref{def consta-circulant}).
As in Example~\ref{exm}(2), ${\rm M}_n(F)$ denotes the $n\times n$
matrix algebra over $F$.
% consisting of all $F$-matrices of degree $n$.
We consider its subset consisting of
all the $\lambda$-consta circulant matrices of degree~$n$:
\begin{align*}
{\rm M}_{n\dash\lambda\dash{\rm circ}}(F)=
\big\{\,a(P_\lambda)\;\big|\;a(X)\in\RL\,\big\},
\end{align*}
where $P_\lambda$ is defined in Eq.\eqref{eq def P_lambda}
and $a(P_\lambda)$ is defined in Eq.\eqref{lambda-circulant matrix}.

\begin{lemma} \label{lem n circ}
With the notation as above.
${\rm M}_{n\dash\lambda\dash{\rm circ}}(F)$
is a subalgebra of ${\rm M}_n(F)$ and
the following is an algebra isomorphism:
\begin{align} \label{R_n,lambda cong ...}
\nu:\; {\RL}\; \mathop{\longrightarrow} %\limits^{\cong}
\; {\rm M}_{n\dash\lambda\dash{\rm circ}}(F),
 \quad a(X)\,\longmapsto\, a(P_{\!\lambda}).
\end{align}
\end{lemma}
\begin{proof}
The following is obviously an $F$-algebra homomorphism:
\begin{align*} %\label{eq FX to Mn}
 \delta:~~ F[X]\to {\rm M}_n(F),~~~ f(X)\mapsto f(P_{\!\lambda}).
\end{align*}
Since $P_{\!\lambda}^n=\lambda E_n$,
the kernel of the homomorphism $\delta$ %Eq.\eqref{eq FX to Mn}
is the ideal $\langle X^n-\lambda\rangle$ of~$F[X]$ generated by $X^n-\lambda$.
Note that the quotient algebra $F[X]/\langle X^n-\lambda\rangle={\RL}$, see Eq.\eqref{eq R_lambda}.
By Homomorphism Theorem, the homomorphism $\delta$ %Eq.\eqref{eq FX to Mn}
induces an injective homomorphism:
\begin{align*}
 \hat \delta:~~ \RL=F[X]/\langle X^n-\lambda\rangle~\to~ {\rm M}_n(F),~~~ a(X)~\mapsto~ a(P_{\!\lambda}).
\end{align*}
The image of this homomorphism $\hat \delta$ is exactly
${\rm M}_{n\dash\lambda\dash{\rm circ}}(F)$,
cf.~Definition \ref{def consta-circulant}.
Thus ${\rm M}_{n\dash\lambda\dash{\rm circ}}(F)$
is a subalgebra of ${\rm M}_n(F)$
and the homomorphism $\hat \delta$ induces
the algebra isomorphism $\nu$ in Eq.\eqref{R_n,lambda cong ...}.
\end{proof}

In the rest of this subsection we assume that $\lambda^{1+p^h}=1$
i.e., $\lambda=\lambda^{-p^{\ell-h}}$, cf. Remark~\ref{lambda^1+p^h}.
By Eq.\eqref{eq def *h}, we have
the following map
\begin{align*} %\label{eq anti nu}
\tilde\tau:~ {\rm M}_{n}(F) \,\longrightarrow\, {\rm M}_{n}(F), ~~~
A\,\longmapsto\, A^{*h};
\end{align*}
and by Eq.\eqref{eq *h satisfy}, for any $A,B\in {\rm M}_{n}(F)$ and $\alpha\in F$,
\begin{align*}
(AB)^{*h}=B^{*h} A^{*h}, ~~
(A+B)^{*h}=A^{*h}+B^{*h}, ~~
 (\alpha A)^{*h}=\alpha^{p^h}\!A^{*h}.
\end{align*}
So $\tilde\tau$
is a $p^h$-anti-automorphism of the matrix algebra ${\rm M}_n(F)$.

\begin{lemma}
Keep the notation as above.
Assume that $\lambda^{1+p^{h}}=1$. %(i.e., $\lambda=\lambda^{-p^{\ell-h}}$).
Then $\tilde\tau\big({\rm M}_{n\dash\lambda\dash{\rm circ}}(F)\big)
\subseteq {\rm M}_{n\dash\lambda\dash{\rm circ}}(F)$,
and the restricted map
$\tau=\tilde\tau|_{{\rm M}_{n\dash\lambda\dash{\rm circ}}(F)}$
as follows is a $p^h$-automorphism of the $F$-algebra
${\rm M}_{n\dash\lambda\dash{\rm circ}}(F)$:
%(where the operator $*h$ is defined in Eq.\eqref{eq def *h})
\begin{align} \label{*h is an auto}
\tau:\; {\rm M}_{n\dash\lambda\dash{\rm circ}}(F)
 \;\mathop{\longrightarrow} %\limits^{\cong}
 \; {\rm M}_{n\dash\lambda\dash{\rm circ}}(F), \quad
  a(P_\lambda)\,\longmapsto\, \big(a(P_\lambda)\big)^{*h}.
\end{align}
\end{lemma}
\begin{proof}
Since $\lambda=\lambda^{-p^{\ell-h}}$,
by Eq.\eqref{G P_lambda^-1} we deduce that
\begin{align} \label{GG P_lambda^-1}
 P_\lambda^{\,*h}=P_{\!\lambda}^{-1}=P_{\!\lambda}^{nt-1}
\in{\rm M}_{n\dash\lambda\dash{\rm circ}}(F).
\end{align}
By Eq.(\ref{R_n,lambda cong ...}),
any $a(P_\lambda)\in{\rm M}_{n\dash\lambda\dash{\rm circ}}(F)$
is associated with $a(X)\in{\RL}$,
where $a(X)=\sum_{i=0}^{n-1}a_iX^i\in\RL$, so
$$
\big(a(P_\lambda)\big)^{*h}
 =\Big(\sum_{i=0}^{n-1}a_{i}P_\lambda^{\,i}\Big) ^{*h}
 =\sum_{i=0}^{n-1}a_{i}^{p^h}(P_\lambda^{\,*h})^i
  \in {\rm M}_{n\dash\lambda\dash{\rm circ}}(F).
$$
Thus $\tilde\tau\big({\rm M}_{n\dash\lambda\dash{\rm circ}}(F)\big)
\subseteq {\rm M}_{n\dash\lambda\dash{\rm circ}}(F)$.
Restricting the $p^h$-anti-automorphism $\tilde\tau$
to ${\rm M}_{n\dash\lambda\dash{\rm circ}}(F)$, we
get the $p^h$-anti-automorphism Eq.\eqref{*h is an auto},
which is in fact a $p^h$-automorphism because
 ${\rm M}_{n\dash\lambda\dash{\rm circ}}(F)\cong \RL$
is a commutative algebra.
\end{proof}

Next, we introduce an operator ``$*$'' on $\RL$,
which is the key to obtaining the necessary and sufficient conditions for
$2$-quasi $\lambda$-constacyclic codes being Galois $p^h$-self-dual.
With the isomorphism $\nu$ in Eq.\eqref{R_n,lambda cong ...},
inspiring by Eq.\eqref{*h is an auto} and Eq.\eqref{GG P_lambda^-1},
for $a(X)=\sum_{i=0}^{n-1}a_iX^i\in\RL$ we define
\begin{align} \label{def *}
a^*(X)=a^{(p^h)}(X^{nt-1})~({\rm mod}~X^n-\lambda),
\end{align}
where $a^{(p^h)}(X)=\sum_{i=0}^{n-1}a_i^{p^h}X^i$, cf. Example~\ref{exm}(1).

\begin{lemma} \label{lem tau and *}
Assume that $\lambda^{1+p^{h}}=1$. %(i.e., $\lambda=\lambda^{-p^{\ell-h}}$).
Let $a^*(X)$, $\nu$ and $\tau$ be as in  Eq.\eqref{def *},
Eq.\eqref{R_n,lambda cong ...} and Eq.\eqref{*h is an auto}
respectively.
Then
\begin{align} \label{nu * =}
a^*(X)=\nu^{-1}\tau\nu\big(a(X)\big), ~~~~ \forall\; a(X)\in\RL;
\end{align}
and the following map is a $p^h$-automorphism of the algebra $\RL$:
\begin{align} \label{* is an auto}
 *: ~ \RL\;\mathop{\longrightarrow}%\limits^{\cong}
\; \RL, ~~~  a(X)\,\longmapsto\, a^*(X)\,.
\end{align}
\end{lemma}

\begin{proof}
 For $a(X)\in\RL$,
by the definition of $a^*(X)$ in Eq.\eqref{def *}
and by Eq.\eqref{GG P_lambda^-1}, %we have
\begin{align*}
a^*(P_{\!\lambda})
&=\sum_{i=0}^{n-1}a_i^{p^h}(P_{\!\lambda}^{nt-1})^i
 =\sum_{i=0}^{n-1}a_i^{p^h}(P_{\!\lambda}^{*h})^i
 =\sum_{i=0}^{n-1}a_i^{p^h}(P_{\!\lambda}^{i})^{*h} \\
&=\sum_{i=0}^{n-1}(a_i P_{\!\lambda}^{i})^{*h}
 =\Big(\sum_{i=0}^{n-1} a_i P_{\!\lambda}^{i}\Big)^{*h}
 =\big(a(P_{\!\lambda})\big)^{*h}.
\end{align*}
That is,
$\nu\big(a^*(X)\big)=a^*(P_{\!\lambda})=\big(a(P_{\!\lambda})\big)^{*h}
 =\tau\nu\big(a(X)\big)$. So
Eq.\eqref{nu * =} holds; equivalently, the following diagram is commutative:
\begin{align} \label{eq * eta tau}
\begin{array}{ccc}
\RL & \mathop{\longrightarrow}\limits^{\nu}
  & {\rm M}_{n\dash\lambda\dash{\rm circ}}(F) \\[5pt]
 *\,\big\downarrow~ && \big\downarrow\,\tau\\[5pt]
\RL & \mathop{\longrightarrow}\limits_{\nu}
  & {\rm M}_{n\dash\lambda\dash{\rm circ}}(F)
\end{array}
\end{align}
Because $\tau$ is a $p^h$-automorphism and both $\nu$ and $\nu^{-1}$
are algebra isomorphisms, by Eq.\eqref{nu * =}
we see that Eq.\eqref{* is an auto} is a $p^h$-automorphism of $\RL$.
\end{proof}

%\begin{definition}\label{def type 1}
%For any $g\in{\RL}$, we call
%$$C_{1,g}=\{(u,ug)~|~u\in{{\RL}}\}\leq \RL^2$$
%is a 2-quasi constacyclic code of type I.
%\end{definition}

Recall that %the $\lambda$-constacyclic code $C_a$ generated by $a\in\RL$ and
the $2$-quasi $\lambda$-constacyclic code $C_{a,a'}$ generated by $(a,a')\in\RL^2$
has been defined in Remark~\ref{rk C_a}.
For the $\RL$-submodules of $\RL^2$ generated by one element,
we have the following Galois self-duality criteria.

%\begin{lemma} \label{lem p-hC_a,a'}
%Assume that $\lambda^{1+p^h}=1$.

%{\bf(1)}\; $\langle C_{a}, C_{a}\rangle_{h}=0$ if and only if
%$aa^*=0$.

%{\bf(2)}\;
%$\langle C_{a,a'}, C_{a,a'}\rangle_{h}=0$ if and only if
%$aa^*+a'{a'}^*=0$.
%\end{lemma}
%\begin{proof}
%Observe that by Lemma \ref{lem C_a,a'}, $C_{a,a'}$ is linearly generated by the rows of the
%matrix $\big(a(P_{\!\lambda}),\,a'(P_{\!\lambda})\big)$.
%Thus $\langle C_{a,a'}, C_{a,a'}\rangle_{h}=0$ if and only if
%$$
%\big(a(P_{\!\lambda}),\,a'(P_{\!\lambda})\big)\cdot
%\big(a(P_{\!\lambda}),\,a'(P_{\!\lambda})\big)^{*h}
%  =a(P_{\!\lambda})\cdot a(P_{\!\lambda})^{*h}
% + a'(P_{\!\lambda})\cdot a'(P_{\!\lambda})^{*h}=0.
%$$
%It follows from Lemma \ref{lem tau and *}
%(cf. %the commutative diagram
%Eq.\eqref{eq * eta tau}) that
%$\langle C_{a,a'}, C_{a,a'}\rangle_{h}=0$ if and only if
%$aa^*+a'{a'}^*=0$.
%And (1) is proved in a similar way.
%\end{proof}

\begin{lemma} \label{lem C1,g}
Assume that $\lambda^{1+p^h}=1$.

{\bf(1)} The $2$-quasi $\lambda$-constacyclic code
$C_{a,a'}$ is Galois $p^h$-self-dual if and only if
$aa^*+a'a'^*\!=\!0$ and the rate ${\rm R}(C_{a,a'})=1/2$.

{\bf(2)} The $2$-quasi $\lambda$-constacyclic code $C_{1,g}$
 is Galois $p^h$-self-dual if and only if $gg^*=-1$.
\end{lemma}

\begin{proof}
{\rm (1)}~The $C_{a,a'}$ is Galois $p^h$-self-dual if and only if
$\langle C_{a,a'}, C_{a,a'}\rangle_{h}=0$ and $\dim_{F}C_{a,a'}=n$.
What remains is to show that
\begin{align}\label{eq p-hC_a,a'}
  \langle C_{a,a'}, C_{a,a'}\rangle_{h}=0 \;\iff\; aa^*+a'{a'}^*=0.
\end{align}
Observe that by Lemma \ref{lem C_a,a'},
$C_{a,a'}$ is linearly generated by the rows of the
matrix $\big(a(P_{\!\lambda}),\,a'(P_{\!\lambda})\big)$.
Thus $\langle C_{a,a'}, C_{a,a'}\rangle_{h}=0$ if and only if
$$
\big(a(P_{\!\lambda}),\,a'(P_{\!\lambda})\big)\cdot
\big(a(P_{\!\lambda}),\,a'(P_{\!\lambda})\big)^{*h}
  =a(P_{\!\lambda})\cdot a(P_{\!\lambda})^{*h}
 + a'(P_{\!\lambda})\cdot a'(P_{\!\lambda})^{*h}=0.
$$
Thus, Eq.(\ref{eq p-hC_a,a'}) follows from Lemma \ref{lem tau and *}
(cf.  Eq.\eqref{eq * eta tau}) immediately.

{\rm (2)}~
By Corollary \ref{Cor.C1,a}, $C_{1,g}$ has the generating matrix
$\big(E_n,\,g(P_{\!\lambda})\big)$.
In particular, $\dim_F C_{1,g}=n$,
i.e., ${\rm R}(C_{1,g})=1/2$.
 Similarly to the proof of $(1)$, we have
$\langle C_{1,g}, C_{1,g}\rangle_{h}=0$ if and only if $1+gg^*=0$.
\end{proof}

In the semisimple case,
extending \cite[Theorem 4.2]{FZ23}, we have the following theorem
to characterize the Galois self-dual $2$-quasi $\lambda$-constacyclic codes.
If $\gcd(n,q)=1$, by Corollary~\ref{lem structure C again}
any $\RL$-submodule $C$ of $\RL^2$ can be written as
\begin{align} \label{any 2-QCC}
 C=C_{a,0}\oplus C_{0,a'}\oplus C_{b,bg},
\end{align}
where $a,a',b\in{\RL}$ with $ab=a'b=0$ and $g\in{C_b^{\times}}$.

\begin{theorem}\label{th lambda eq iff}
Assume that $\lambda^{1+p^h}=1$ and $\gcd(n,q)=1$.
Let $C$ in Eq.\eqref{any 2-QCC} be any $2$-quasi $\lambda$-constacyclic code.
Then $C$ is Galois $p^h$-self-dual if and only if
the following two hold:

{\rm (1)}~
 $aa^*=a'{a'}^*=ab^*=a^*b=a'b^*=a'^*b=bb^*(1+gg^*)=0;$

{\rm (2)}~ $\dim_{F}C=n$.
\end{theorem}

\begin{proof}
The $C$ is Galois $p^h$-self-dual if and only if
$\langle C,C\rangle_h=0$ and $\dim_F C=n$.
Note that the inner product $\langle -,-\rangle_h$ is linear for the first
variable,  and it is $p^h$-linear for the second variable,
but it is not symmetric in general.
So $\langle C,C\rangle_h=0$ is equivalent to the following
\begin{align*}
&\langle C_{a,0}, C_{0,a'}\rangle_h=\langle C_{0, a'}, C_{a,0}\rangle_h=0, \\
& \langle C_{a,0}, C_{a,0}\rangle_h=\langle C_{0,a'}, C_{0,a'}\rangle_h=0, \\
& \langle C_{a,0}, C_{b,bg}\rangle_h=\langle C_{b,bg}, C_{a,0}\rangle_h=0, \\
& \langle C_{0,a'}, C_{b,bg}\rangle_h=\langle C_{b,bg}, C_{0,a'}\rangle_h=0,  \\
&\langle C_{b,bg}, C_{b,bg}\rangle_h=0 .
\end{align*}
The first line holds obviously. By Eq.(\ref{eq p-hC_a,a'}),
the second and the third lines are equivalent to that
$aa^*=a'a'^*=0$ and $ab^*=a^*b=0$, respectively.
The last line is equivalent to that $bb^*(1+gg^*)=0$.

Turn to the forth line which is equivalent to $a'b^*g^*=a'^*bg=0$.
Note that $C_b=\RL\,e_b$ is a ring with identity $e_b$ which is an idempotent,
cf. Remark~\ref{rk R semisimple}(1), hence
``$g\in C_b^\times$'' implies that $gg'=e_b$ for a $g'\in C_b$.
So $a'^*b=a'^*be_b=a'^*bgg'=0g'=0$. Similarly,
$a'b^*=a'b^*e_b^*=a'b^*g^*g'^*=0g'^*=0$.

The theorem is proved.
\end{proof}

\section{Hermitian self-dual $2$-quasi-cyclic codes}
\label{section H self-dual 2-Q cyclic}

In this section, we always assume that $\ell$ is even and $h=\ell/2$,
and ${\gcd(n,q)=1}$.
The map $\sigma_{\ell/2}: F\to F$, $\alpha\mapsto \alpha^{p^{\ell/2}}$, is
a Galois automorphism of order $2$, and
\begin{align*}
\langle{\bf a},{\bf a'}\rangle_{\ell/2}
 =\sum_{i=0}^{n-1}a_i {a'_i}^{p^{\ell/2}},
~~~~ \forall\, {\bf a}=(a_0,\cdots,a_{n-1}),\,{\bf a'}=(a'_0,\cdots,a'_{n-1})\in F^n,
%{\bf a},{\bf b}\in F^n,
\end{align*}
is the Hermitian inner product on $F^{n}$.
In this section we consider $\R1=\R1_1=F[X]/\langle X^n-1\rangle$
and $\R1^2=\R1\times \R1$, cf. Eq.\eqref{eq R=R_1};
and prove that the Hermitian self-dual $2$-quasi-cyclic codes
are asymptotically good.

\subsection{The operator ``$*$'' on $\R1=F[X]/\langle X^n-1\rangle$}

Since $1^{1+p^{\ell/2}} =1$,
the results in Subsection~\ref{sub lambda=}
can be quoted freely for $\R1$ and $\R1^2$.
In particular, the operator ``$*$'' in Lemma~\ref{lem tau and *}
is a $p^{\ell/2}$-automorphism of~$\R1$:
\begin{align} \label{* on cyclic}
 *:~~ \R1\; \longrightarrow\;\R1, ~~~ a(X)\;\longmapsto \; a^*(X),
\end{align}
where $a(X)=\sum_{i=0}^{n-1}a_iX^i$ and
$a^*(X)=a^{(p^{\ell/2})}(X^{n-1})~({\rm mod}~X^n-1)$, i.e.,
$a^*(X)
=
\sum_{i=0}^{n-1}a_i^{p^{\ell/2}}(X^{n-1})^i~({\rm mod}~X^n-1)$,
cf. Eq.\eqref{def *}.
By Eq.\eqref{eq R to FG}, we have the identification:
\begin{align*}
  \R1\cong FG, ~~~~ \sum_{i=0}^{n-1}a_iX^i~ \mapsto~ \sum_{i=0}^{n-1}a_ix^i,
\end{align*}
where $G=\langle\,x\,|\,x^n=1\rangle$ is the cyclic group of order $n$
and $FG$ is the cyclic group algebra.
So the symbol ``$X$'' can be identified with the element $x$ of the group~$G$,
and we can write $a(x)=\sum_{i=0}^{n-1}a_ix^i\in\R1$, and the expressions $x^{-1}$,
$a(x^{-1})$ etc. make sense. Hence $x^{n-1}=x^{-1}$ and
\begin{align} \label{eq a(x^-1)}
 a^*(x)=\sum_{i=0}^{n-1}a_i^{p^{\ell/2}}(x^{-1})^i
={a}^{(p^{\ell/2})}(x^{-1}), ~~~~\forall\;a(x)\in\R1.
\end{align}

\begin{lemma} \label{lem * order 2}
The $p^{\ell/2}$-automorphism ``$*$'' of $\R1$ in
Eq.\eqref{* on cyclic}
is of order 2.
\end{lemma}

\begin{proof}
Since $x^{n-1}=x^{-1}$, for any $a(x)\in\R1$ by Eq.\eqref{eq a(x^-1)} we have
\begin{align*}
\Big(\sum_{i=0}^{n-1} a_ix^i\Big)^{**}
=\Big(\sum_{i=0}^{n-1} a_i^{p^{\ell/2}}(x^{-1})^i\Big)^*
=\sum_{i=0}^{n-1} (a_i^{p^{\ell/2}})^{p^{\ell/2}} ((x^{-1})^{-1})^i
=\sum_{i=0}^{n-1} a_ix^i.
\end{align*}
Thus the order of the operator ``$*$'' equals $1$ or $2$.
There is an $\alpha\in F$ such that $\alpha^{p^{\ell/2}}\!\neq\alpha$.
Then in $\R1$ we have
$(\alpha 1)^*\neq\alpha 1$.
So the order of the operator ``$*$'' equals $2$.
\end{proof}

\begin{corollary} \label{cor * order 2}
If $C$ is a non-zero ideal of $\R1$ which is invariant
by the operator~``$*$'' (i.e., $C^*=C$),
then the restriction of the operator ``$*$'' to $C$
induces a $p^{\ell/2}$-automorphism of~$C$ of order $2$.
\end{corollary}

\begin{proof}
By Lemma \ref{lem * order 2},
the restriction of ``$*$'' to $C$ is of order $1$ or $2$.
Let $c\in C$ with $c\neq 0$. If $c^*\neq c$,
then the restriction of ``$*$'' to $C$ is not the identity.
Otherwise,  $c^*=c$; there is an $\alpha\in F$
with $\alpha^{p^{\ell/2}}\!\neq \alpha$; so
$(\alpha c)^*=\alpha^{p^{\ell/2}}\!c^*=\alpha^{p^{\ell/2}}c\neq \alpha c$.
In conclusion, the restriction of ``$*$'' to $C$ is of order $2$.
\end{proof}

Note that ``$\gcd(n,q)=1$'' is assumed in this section.
Recall from Eq.\eqref{eq e_0 ...} that
$e_0=\frac{1}{n}\sum_{i=0}^{n-1}x^i$, $e_1,\dots, e_m$
are all primitive idempotents of~$\R1$.
Then the $p^h$-automorphism ``$*$'' permutes the primitive idempotents,
i.e., every $e_i^*$ is still a primitive idempotent.
Note that $e_0^*=e_0$. We can reorder the other primitive idempotents
\begin{align} \label{eq idempotents reorder}
 e_1, \dots, e_r,\; e_{r+1}, e_{r+1}^*, \dots, e_{r+s}, e_{r+s}^*,
\end{align}
such that $e_i^*=e_i$ for $i=1,\dots,r$,
$e_{r+j}^*\ne e_{r+j}$ but $e_{r+j}^{**}=e_{r+j}$ for $j=1,\dots,s$.
%and $1+r+2s=m$.

\begin{remark}\label{rk restriction *} \rm
{\rm (1)} For $i=0,1, \dots, r$, we get $(\R1 e_i)^{*}=\R1 e_i^*=\R1 e_i$. %$e_i^*=e_i$.
The restriction of the map ``$*$'' in Eq.(\ref{* on cyclic})
to $\R1 e_i$ induces a $p^{\ell/2}$-automorphism of $\R1 e_i$
of order 2 (cf. Corollary \ref{cor * order 2}) as follows
\begin{equation*} %\label{}
 *\!|_{\R1 e_i}:~ \R1 e_i \rightarrow \R1 e_i, ~~~ a\mapsto a^*.
\end{equation*}
Let $d_i=\dim_F\R1 e_i$. By Eq.(\ref{eq R=...}),
the ideal $\R1 e_i$ is a field extension over $F$ with
 cardinality $|\R1 e_i|=p^{d_i\ell}$.
So $a^*=a^{p^{d_i\ell/2}}$ for $a\in\R1 e_i$.

{\rm (2)} For $j=1, \dots, s$, $(\R1 e_{r+j})^{*}=\R1 e_{r+j}^{*}\neq \R1 e_{r+j}$,
and the restriction of the map ``$*$'' in Eq.(\ref{* on cyclic})
to $\R1 e_{r+j}$ induces a $p^{\ell/2}$-isomorphism as follows
\begin{equation*}
 *\!|_{\R1 e_{r+j}}:~
\R1 e_{r+j} \rightarrow \R1 e_{r+j}^{*}, ~~~ a\mapsto a^*.
\end{equation*}
Let $d_{r+j}=\dim_F\R1 e_{r+j}$,
then $\dim_F \R1 e_{r+j}^* =\dim_F \R1 e_{r+j}=d_{r+j}$.
Denote $\widehat e_{r+j}=e_{r+j}+e_{r+j}^*$,
by Eq.(\ref{eq R=...}) we have that
\begin{align*}
& \R1\widehat e_{r+j}=\R1 e_{r+j}\oplus\R1 e_{r+j}^*
=\{a'+a''\,|\,a'\in\R1 e_{r+j},\,a''\in\R1e_{r+j}^*\}, \\
& \dim_F \R1\widehat e_{r+j} =2d_{r+j}.
\end{align*}
It is easy to check that
$\widehat e_{r+j}^{\,*}=(e_{r+j}+e_{r+j}^*)^{*}=e_{r+j}^*+e_{r+j}=\widehat e_{r+j}$.
So $(\R1 \widehat e_{r+j})^{*}=\R1 \widehat e_{r+j}^{\;*}=\R1 \widehat e_{r+j}$.
The restriction of the map ``$*$'' in Eq.(\ref{* on cyclic})
to $\R1\widehat e_{r+j}$ induces a
$p^{\ell/2}$-automorphism of $\R1 \widehat e_{r+j}$ of order 2
(cf. Corollary \ref{cor * order 2}) as follows
\begin{equation} \label{eq *|_Re_i+j}
 *\!|_{\R1 \widehat e_{r+j}}:~
\R1 \widehat e_{r+j} \rightarrow \R1 \widehat e_{r+j}, ~~~
a'+a''\,\mapsto\, a''^*+a'^*.
\end{equation}
\end{remark}

For convenience, in the following we denote $\widehat e_i=e_i$ for $i=0,1,\dots, r$.
Then $\widehat e_{i}^{\,*}=\widehat e_{i}$ for $i=0,1,\dots, r+s$, and
\begin{align} \label{eq 1=e_0+...}
 1=\widehat e_0+\widehat e_1+\dots+\widehat e_{r+s}, ~~~~~
 \widehat e_i\widehat e_j
  =\begin{cases} \widehat e_i, & i=j;\\ 0, &i\ne j. \end{cases}
\end{align}
Thus, $\R1$ can be rewritten as  % , cf. Eq.(\ref{eq R=...}),
\begin{align} \label{eq R1=}
\R1=\R1 \widehat e_0\oplus \R1 \widehat e_1\oplus\dots\oplus\R1 \widehat e_r \oplus
 \R1 \widehat e_{r+1}\oplus\dots\oplus \R1\widehat e_{r+s}.
\end{align}
Any $a\in \R1$ is decomposed into
\begin{align}  \label{eq a=a_0+...}
 a=a_0+a_1+\cdots+ a_{r+s}, ~~~~~
 \mbox{where}~~ a_i=a\widehat e_i, ~ i=0,1,\dots, r+s.
\end{align}
The $a_i=a\widehat e_i$ is called the {\em $\widehat e_i$-component} of $a$.
For $a=\sum_{i=0}^{r+s}a_i, b=\sum_{i=0}^{r+s}b_i\in\R1$,
by Eq.\eqref{eq 1=e_0+...} it is trivial to check that
\begin{align} \label{eq ab=...}
ab=a_0b_0+a_1b_1+\cdots+a_{r+s}b_{r+s}.
\end{align}
Keep the notation in Remark~\ref{rk restriction *},
we have $d_{i}=\dim_F\R1 \widehat e_{i}$ for $i=0,1,\dots, r$;
and $2d_{i}=\dim_F\R1 \widehat e_{i}$ for $i=r+1,\dots, r+s$. Thus,
\begin{align} \label{eq d_0+...}
 \dim_{F}{\R1}%=\sum_{i=0} ^{r+s}\dim_F\R1 \widehat e_{i}
 = d_0+d_1+\dots+d_r+2d_{r+1}+\dots+2d_{r+s}=n,
\end{align}
where $d_0=1$ and $d_i\ge \mu(n)$ for $i=1,\dots,r+s$; cf. Eq.\eqref{eq mu(n)=...}.
%we have

\subsection{A class of Hermitian self-dual $2$-quasi-cyclic codes}
\label{H self-dual 2-Q}

Recall that the $2$-quasi-cyclic code $C_{1,g}$ of $\R1^2$
(defined in Remark~\ref{rk C_a})
is Hermitian self-dual if and only if $gg^*=-1$, see Lemma~\ref{lem C1,g}.
So we denote
\begin{align} \label{eq def D}
{\cal D}=\big\{\, g\;\big|\;g\in\R1,\; gg^*=-1\,\big\}.
\end{align}
Any $g\in{\cal D}$ corresponds to a Hermitian self-dual
$2$-quasi-cyclic code $C_{1,g}$.
%For $i=0,1,\dots,r+s$, we set
Set
\begin{align} \label{eq D_i=...}
{\cal D}_i=\{\,z\;|\; z\in\R1\widehat e_i,\;zz^*=-\widehat e_i\}, ~~~~~
i=0,1,\dots,r+s.
\end{align}
By Eq.\eqref{eq 1=e_0+...}, Eq.\eqref{eq a=a_0+...} and  Eq.\eqref{eq ab=...},
\begin{align}\label{eq g in D}
 g=g_0+g_1+\dots+g_{r+s}\in{\cal D}~\iff ~
g_i\in{\cal D}_i, ~ \forall\; i=0,1,\dots, r+s.
\end{align}

\begin{lemma} \label{lem |D_i|=...}
{\bf(1)}
If $0\le i\le r$ (i.e. $\widehat e_i=e_i=e_i^*$),
then $|{\cal D}_i|=p^{d_i\ell/2}+1$.

{\bf(2)}
If $r< i\le r+s$ (i.e. $e_i\neq e_i^*$ and $\widehat e_i=e_i+e_i^*$),
then $|{\cal D}_i|=p^{d_i\ell}-1$.
\end{lemma}
\begin{proof}
{\rm(1)}.
% $0\le i\le r$, i.e. $\widehat e_i=e_i=e_i^*$.
By Remark~\ref{rk restriction *}(1), $\R1\widehat e_i$ is the field with
$|\R1\widehat e_i|=p^{d_i\ell}$. For $z\in\R1\widehat e_i$,
$zz^*=zz^{p^{d_i\ell/2}}\!=z^{ p^{d_i\ell/2}+1}$. So,
$zz^*=-\widehat e_i$ if and only if
${z^{p^{d_i\ell/2}+1}+\widehat e_i=0}$.
Hence~$|{\cal D}_i|$ equals the number of the roots in $F_{p^{d_i\ell}}$
of the $F_{p^{d_i\ell}}$-polynomial $X^{p^{d_i\ell/2}+1}+1$,
where $F_{p^{d_i\ell}}$ denotes the finite field with cardinality $p^{d_i\ell}$.

If $p=2$, then $1=-1$. The order of the multiplicative group
$F_{p^{d_i\ell}}^\times$ is
\begin{align}\label{eq F times}
 \big|F_{p^{d_i\ell}}^\times\big|=p^{d_i\ell}-1
 =\big(p^{d_i\ell/2}-1\big)\big(p^{d_i\ell/2}+1\big).
\end{align}
Thus the multiplication group $F_{p^{d_i\ell}}^\times$
has a subgroup $H$ of order $p^{d_i\ell/2}+1$,
and all elements of $H$ are roots of
the polynomial $X^{p^{d_i\ell/2}+1}+1$.
Hence $|{\cal D}_i|=p^{d_i\ell/2}+1$.

Otherwise $p$ is odd, then $p^{d_i\ell/2}-1$ is even.
By Eq.(\ref{eq F times}),
we get that $2(p^{d_i\ell/2}+1)\,\big|\,|F_{p^{d_i\ell}}^\times|$.
So $F_{p^{d_i\ell}}^\times$ has a subgroup
$H$ of order $2(p^{d_i\ell/2}+1)$. The elements of $H$ are
just all roots of the polynomial $X^{2(p^{d_i\ell/2}+1)}-1$.
Since
$$
 X^{2(p^{d_i\ell/2}+1)}-1=(X^{p^{d_i\ell/2}+1}-1)(X^{p^{d_i\ell/2}+1}+1),
$$
all roots of the polynomial $X^{p^{d_i\ell/2}+1}+1$ are inside~$F_{p^{d_i\ell}}$.
Thus, $|{\cal D}_i|=p^{d_i\ell/2}+1$.

{\rm(2)}.
For $r< i\leq r+s$, $e_i\neq e_i^*$ and $\widehat e_i=e_i+e_i^*$.
By Remark~\ref{rk restriction *}(2),
 $\R1\widehat e_i=\R1 e_i\oplus\R1 e_i^*$.
For $z=z'+z''\in\R1 e_i\oplus\R1 e_i^*$
with $z'\in\R1 e_i$ and $z''\in\R1 e_i^*$,
by Eq.\eqref{eq *|_Re_i+j} we have $z^*=z''^*+z'^*$, and so
$zz^*=(z'+z'')(z''^*+z'^*)=z'z''^*+z'^*z''$.
It follows that: $zz^*=-\widehat e_i=-e_i-e_i^*$ if and only if
$z'z''^*=-e_i$ and $z'^*z''=-e_i^*$. We take
$z'\in(\R1 e_i)^\times$, then $z''^*=-z'^{-1}$
(where $z'^{-1}$ is the inverse of $z'$ in $\R1 e_i$, not in $\R1$),
hence $z''=z''^{**}=-(z'^*)^{-1}$
is uniquely determined. Thus,
\begin{align} \label{eq zz^*=...}
 zz^*=-\widehat e_i ~\iff~
 z = z' -(z'^*)^{-1}~~ \mbox{for a}~ z'\in(\R1 e_i)^\times.
\end{align}
Since $\R1 e_i$ is a field with cardinality $p^{d_i\ell}$,
 $|{\cal D}_i|=|(\R1 e_i)^\times|=p^{d_i\ell}-1$.
\end{proof}

For any subset $\omega\subseteq\{0,1,\dots,r\!+\!s\}$,
refining the notation in Eq.\eqref{eq R1=} and in Remark~\ref{rk restriction *},
we define an ideal $A_\omega$ of $\R1$ and an integer~$d_{\omega}$ as follows
\begin{align*} %\label{eq A_omega}
\textstyle
A_{\omega}=\bigoplus\limits_{i\in\omega}\R1 \widehat e_i, ~~~~
d_\omega=\dim_F A_\omega.
\end{align*}
Obviously, $\omega$ can be written as a disjoint union:
\begin{align*}
\omega=\omega'\cup\omega'',~~~
 \mbox{where }
 \omega'\!=\!\omega\cap\{0,1,\dots, r\},~
 \omega''\!=\!\omega\cap\{r\!+\!1,\dots, r\!+\!s\}.
\end{align*}
Similarly to Eq.\eqref{eq d_0+...}, we have
\begin{align} \label{eq A_omega=}
 A_\omega=A_{\omega'}\oplus A_{\omega''}, ~~~~
 %\dim_F A_{\omega}
d_\omega=\sum_{i\in\omega'} d_i + 2\sum_{i\in\omega''} d_i .
\end{align}

It is known that (\cite[Lemma 4.7]{LF22})
for integers $ k_{1}, \dots, k_{v}$,
if $k_i\geq \log_{p}(v)$ for $i=1, \dots, v$,
then
\begin{align} \label{eq LF22 4.7}
\begin{array}{l}
 (p^{k_1}-1) \dots (p^{k_v}-1)\geq p^{k_1+\dots+k_v -2};\\[3pt]
 (p^{k_1}+1) \dots (p^{k_v}+1)\leq p^{k_1+\dots+k_v +2}.
\end{array}
\end{align}

\begin{lemma} \label{lem |D|<}
For a subset $\omega\subseteq\{0,1,\dots,r+s\}$, if $\mu(n)\geq\log_p(n)$
(where $\mu(n)$ is defined in Eq.\eqref{eq mu(n)=...}), then
\begin{align*} %\label{eq |D|<}
 p^{-2}p^{d_\omega\ell/2} \leq
\prod_{i\in\omega}|{\cal D}_i| \leq  p^3 p^{d_\omega\ell/2}.
\end{align*}
\end{lemma}
\begin{proof}
For $1\le i\le r+s$, we deduce that $d_i\ge\mu(n)\ge\log_p(r+s)$ since $n\geq r+s$.
By Lemma~\ref{lem |D_i|=...} we have that
\begin{align*}
\prod_{i\in\omega}|{\cal D}_i|
&=\prod_{i\in\omega'}|{\cal D}_i|\cdot \prod_{i\in\omega''}|{\cal D}_i|
=\prod\limits_{i\in\omega'}(p^{d_i\ell/2}+1)
\cdot\prod\limits_{i\in\omega''}(p^{d_i\ell}-1).
\end{align*}
If $0\notin\omega$, using Eq.(\ref{eq LF22 4.7}) we get
\begin{align*}
\prod_{i\in\omega}|{\cal D}_i|&\ge\prod\limits_{i\in\omega'}(p^{d_i\ell/2}-1)
\cdot\prod\limits_{i\in\omega''}(p^{d_i\ell}-1)
\ge
p^{\big(\sum_{i\in\omega'} d_i\ell/2+\sum_{i\in\omega''}d_i\ell\big)-2}\\
&= p^{-2}p^{\big(\sum_{i\in\omega'} d_i
 +\sum_{i\in\omega''}2d_i \big)\ell/2}=p^{-2}p^{d_{\omega}\ell/2},
\end{align*}
where the last equality follows by Eq.\eqref{eq A_omega=};
and
\begin{align*}
\prod_{i\in\omega}|{\cal D}_i|&\le\prod\limits_{i\in\omega'}(p^{d_i\ell/2}+1)
\cdot\prod\limits_{i\in\omega''}(p^{d_i\ell}+1)
\le
p^{\big(\sum_{i\in\omega'} d_i\ell/2+\sum_{i\in\omega''}d_i\ell\big)+2}\\
&= p^{2}p^{\big(\sum_{i\in\omega'} d_i
 +\sum_{i\in\omega''}2d_i \big)\ell/2}
=p^{2}p^{d_{\omega}\ell/2}.
\end{align*}
That is, if $0\notin\omega$ then
\begin{equation}\label{eq leq D leq}
p^{-2}p^{d_{\omega}\ell/2}  \leq\prod_{i\in\omega}|{\cal D}_i|
 \leq p^{2}p^{d_{\omega}\ell/2}.
\end{equation}
If $0\in\omega$, we set $\tilde\omega=\omega\setminus\{0\}$,
and so $d_{\omega}=d_{0}+d_{\tilde\omega}$,  where $d_0=1$.
By~Lemma~\ref{lem |D_i|=...}(1),
we see that $|{\cal D}_0|=p^{d_0\ell/2}+1$, and
\begin{align*}
\prod_{i\in\omega}|{\cal D}_i|
& =(p^{d_0\ell/2}+1)\prod\limits_{i\in\tilde\omega}|{\cal D}_i|.
\end{align*}
It follows by Eq.(\ref{eq leq D leq}) that
\begin{align*}
\prod_{i\in\omega}|{\cal D}_i|
\ge p^{d_0\ell/2}\cdot p^{-2}p^{d_{\tilde\omega}\ell/2}
&= p^{-2}p^{\big(d_0+d_{\tilde\omega}\big)\ell/2}
 =p^{-2}p^{d_{\omega}\ell/2};
\end{align*}
and that
\begin{align*}
\prod_{i\in\omega}|{\cal D}_i|
\le p^{1+ d_0\ell/2}\cdot p^{2}p^{d_{\tilde\omega}\ell/2}
 = p^{3}p^{\big(d_0+d_{\tilde\omega}\big)\ell/2}
 =p^{3}p^{d_{\omega}\ell/2}.
\end{align*}
Thus the lemma holds.
\end{proof}

By Eq.\eqref{eq g in D}, $|{\cal D}|=\prod_{i=0}^{r+s}|{\cal D}_i|$.
We have the following at once.

\begin{corollary} \label{cor |D|<}
If $\mu(n)\geq\log_p(n)$
(where $\mu(n)$ is defined in Eq.\eqref{eq mu(n)=...}), then
$$
 p^{-2}p^{n\ell/2} \leq |{\cal D}| \leq  p^3 p^{n\ell/2}.
\eqno\qed
$$
\end{corollary}

Any element $a\in\R1$ can be written as
$a=a_{0}\hat e_{0}+a_{1}\hat e_{1}+\dots+a_{r+s}\hat e_{r+s}$
(cf.~Eq.\eqref{eq a=a_0+...}).
We denote
\begin{align*}
\omega_a=\big\{ i\,\big|\, a_i
 =a\widehat e_i\ne 0\big\}\subseteq\{0,1,\dots,r+s\};~~
 d_a=d_{\omega_a}=\dim_F A_{\omega_a}.%~({\rm cf.}~Eq.(\ref{eq A_omega=})).
\end{align*}
Obviously $\R1 a\subseteq A_{\omega_a}$. For $a,b\in\R1$ it is easy to check that
\begin{align}\label{eq Ra = Rb}
\R1 a =\R1 b ~~ \implies ~~ A_{\omega_a}=A_{\omega_b}~
(\mbox{equivalently, } \omega_a=\omega_b);
\end{align}
but the converse is not true in general.

%For $a,b\in\R1$, i.e, $(a,b)\in\R1^2=\R1\times\R1$, denote
%\begin{align}
% {\cal D}_{(a,b)}=\big\{ g\:\big|\; g\in{\cal D},\;(a,b)\in C_{1,g}\big\}.
%\end{align}

\begin{lemma} \label{lem D_a,b}
For $(a,b)\in\R1^2$, we denote
\begin{align}\label{eq Da,b}
{\cal D}_{(a,b)}=\big\{ g\:\big|\; g\in{\cal D},\;(a,b)\in C_{1,g}\big\},
\end{align}
where $C_{1,g}$ is defined in Remark~\ref{rk C_a}
and ${\cal D}$ is defined in Eq.(\ref{eq def D}).

{\bf(1)}
If ${\cal D}_{(a,b)}\ne\emptyset$, then $\R1 a=\R1 b$, hence $\omega_a=\omega_b$.

{\bf(2)} If $\mu(n)\ge\log_p(n)$, then
$
 |{\cal D}_{(a,b)}| \le p^5p^{(n-d_a)\ell/2}.
$
\end{lemma}
\begin{proof}
{\rm(1)}.~Assume that $g\in{\cal D}_{(a,b)}$, then $gg^*=-1$ (cf. Eq.\eqref{eq def D}) and
$u(1,g)=(a,b)$ for an element $u\in\R1$.
The former equality implies that $g$ is invertible with $g^{-1}=-g^*$.
The latter equality implies that $a=u$ and $b=ug=ag$.
We deduce that $\R1 b=\R1 ag=\R1 ga=\R1 a$ since $g$ is invertible.

{\rm(2)}. By Eq.\eqref{eq a=a_0+...}, we write
$a=a_0+\dots+a_{r+s}$, $b=b_0+\dots+b_{r+s}$ and
$g=g_0+\dots+g_{r+s}$ with $g_i\in{\cal D}_i$.
%As $g$ is invertible, each $g_i$ is invertible in $\R1\widehat e_{i}$.
By Eq.\eqref{eq ab=...},
$$
ag=b~~\iff a_ig_i=b_i~ \mbox{for} ~~ i=0,1,\dots,r+s.
$$
We count the number of such $g_i$ in two cases.

{\bf Case 1}: $a_i=0$, i.e., $i\notin\omega_a$.
Since $\R1 a=\R1 b$, we have that $\R1 b_i=\R1 a_i=0$, hence $b_i=0$.
Then any $g_i\in{\cal D}_i$ satisfies that $a_ig_i=b_i$.

{\bf Case 2}: $a_i\ne 0$, i.e., $i\in\omega_a$. There are two subcases.

Subcase 2.1: $0\le i\le r$. Since $\R1 \widehat e_i$ is a field,
 $a_i$ is invertible in $\R1 \widehat e_i$,
and so there is a unique $z\in \R1 \widehat e_i$ such that
$a_iz=b_i$. We see that there is at most one $g_i\in{\cal D}_i$ such that
$a_ig_i=b_i$.

Subcase 2.2: $r < i\le r+s$.
By Remark~\ref{rk restriction *}(2),
$\R1\widehat e_i=\R1 e_i\oplus\R1 e_i^*$.
We can write $a_i=\alpha'+\alpha''$ and $b_i=\beta'+\beta''$, where
$\alpha',\beta'\in\R1 e_i$ and $\alpha'', \beta''\in\R1 e_i^*$.
Since $a_i\ne 0$, at least one of $\alpha'$ and $\alpha''$ is nonzero.
We may assume that $\alpha'\ne 0$.
Take $g_i=z' -(z'^*)^{-1}\in{\cal D}_i$ as in Eq.\eqref{eq zz^*=...},
then $a_i g_i=b_i$ implies that
$\alpha'z'=\beta'$ in $\R1 e_i$. Hence $z'=\alpha'^{-1}\beta'$ is uniquely determined.

In a word, if $a_i\ne 0$ (i.e., $i\in\omega_a$), then there is at most one
$g_i\in{\cal D}_i$ such that $a_ig_i=b_i$. Thus
\begin{align*}
|{\cal D}_{(a,b)}|\le \prod_{i\notin\omega_a} |{\cal D}_i|
=\prod_{i=0}^{r+s}|{\cal D}_i|\Big/\prod_{i\in\omega_a} |{\cal D}_i|
=|{\cal D}|\Big/\prod_{i\in\omega_a} |{\cal D}_i|.
\end{align*}
By Lemma~\ref{lem |D|<} and Corollary~\ref{cor |D|<},
\begin{align*}
|{\cal D}_{(a,b)}|
\leq p^3 p^{n\ell/2}\Big/p^{-2} p^{d_{\omega_a}\ell/2}
=p^5p^{(n-d_{\omega_a})\ell/2}.
\end{align*}
We are done.
\end{proof}

\subsection{Hermitian self-dual $2$-quasi-cyclic codes are asymptotically good}
\label{H self-dual good}

Keep the notation in Subsection~\ref{H self-dual 2-Q}.
 From now on, let $\delta$ be a real number satisfying that
(where $h_q(\delta)$ is the $q$-entropy function defined in Eq.\eqref{eq def h_q})
\begin{align}\label{eq delta in...}
	\delta\in(0,1-q^{-1}) \quad\mbox{and}\quad
    h_q(\delta)< 1/4.
\end{align}
And set
\begin{align} \label{eq D<delta}
\begin{array}{l}
{\cal D}^{\le\delta}
 =\big\{\,g\;\big|\;g\in{\cal D},\; \Delta(C_{1,g})\le\delta\,\big\};\\[3pt]
\Omega_d=\big\{ A~|~ A ~{\rm is~ an ~ideal ~of ~}\R1 {\rm ~with ~} \dim_{F}A=d\big\}.
\end{array}
\end{align}
Recall that  $\mu(n)$ is defined in Eq.\eqref{eq mu(n)=...},
$(A\times A)^{\le\delta}$ is defined in Eq.\eqref{eq def C^<=} and
${\cal D}_{(a,b)}$ is defined in Eq.(\ref{eq Da,b}).

\begin{lemma} \label{lem D leq delta} %\label{rk D leq delta}
~ ${\cal D}^{\le\delta} \subseteq
 \bigcup_{d=\mu(n)}^{n}~\bigcup_{A\in\Omega_{d}}
 ~\bigcup_{(a,b)\in(A\times A)^{\le\delta}}{\cal D}_{(a,b)}$.
\end{lemma}
\begin{proof}
Assume that $g\in{\cal D}^{\le\delta}$. By Eq.\eqref{eq D<delta},
we have an $(a,b)\in\R1^2$ such that $0<{\rm w}(a,b)\le 2\delta n$
and $(a,b)\in C_{1,g}$, i.e., $g\in{\cal D}_{(a,b)}$.
From Lemma~\ref{lem D_a,b}(1) and Eq.(\ref{eq Ra = Rb}),
we deduce that $A_{\omega _a}=A_{\omega_b}$; hence
$d_a\!=\!\dim_F A_{\omega_a}\!=\!\dim_F A_{\omega_b}=d_b$, and
$(a,b)\in A_{\omega_a}\times A_{\omega_a}$ and $A_{\omega_a}\in\Omega_{d_a}$.
If $A_{\omega_a}=\R1 e_0$, then $0\neq a,b \in\R1 e_0$,
hence  ${\rm w}(a,b)=2n$ (cf. Eq.\eqref{eq R e_0}),
which contradicts that ${\rm w}(a,b)\le 2\delta n$
and $\delta<1-q^{-1}$ (see Eq.\eqref{eq delta in...}).
Thus $A_{\omega_a}\ne\R1 e_0$.
By the definition of $\mu(n)$ in~Eq.\eqref{eq mu(n)=...},
we obtain that $d_a\geq \mu(n)$.
So the lemma is proved.
\end{proof}

\begin{lemma} \label{lem D^leq}
%With notations as above.
If $\frac{1}{4} -h_q(\delta)-\frac{\log_p(n)}{\mu(n)\ell}>0$
and $\mu(n)\geq \log_p(n)$, then
\begin{align*}
|{\cal D}^{\le\delta}|\le p^5 p^{\frac{n\ell}{2}
 - 2\mu(n)\ell\big(\frac{1}{4} -h_q(\delta)-\frac{\log_p(n)}{\mu(n)\ell}\big)}.
\end{align*}
\end{lemma}
\begin{proof}
Note that any ideal $A$ of $\R1$ is a direct sum of some of
$\R1 e_0, \R1 e_1, \dots, \R1 e_{m}$ (cf. Eq.(\ref{eq R=...})).
For $1\le i\le m$, the dimension $\dim_F \R1 e_i\ge\mu(n)$,
where $\mu(n)$ is defined in Eq.(\ref{eq mu(n)=...}).
Suppose that $A\in{\Omega_{d}}$
where $\mu(n)\le d\le n$.
If $\R1 e_0$ is not a direct summand of $A$, then the number of such $A$
is at most $m^{\frac{d}{\mu(n)}}$;
otherwise $\R1 e_0$ is a direct summand of $A$, then the number of such $A$
is at most $m^{\frac{d-1}{\mu(n)}}<m^{\frac{d}{\mu(n)}}$.
Hence
\begin{align} \label{eq Omega}
 |\Omega_d|\le  m^{\frac{d}{\mu(n)}}
\le  n^{\frac{d}{\mu(n)}}.
\end{align}
Applying Lemma~\ref{lem D leq delta} and Lemma~\ref{lem D_a,b}(2), we obtain
\begin{align*}
|{\cal D}^{\le\delta}|
&\le
 \sum_{d=\mu(n)}^{n}~\sum_{A\in\Omega_{d}}~
\sum_{(a,b)\in(A\times A)^{\le\delta}} |{\cal D}_{(a,b)}| \\
&\le
 \sum_{d=\mu(n)}^{n}~\sum_{A\in\Omega_{d}}~
\sum_{(a,b)\in(A\times A)^{\le\delta}} p^5 p^{(n-d)\ell/2},
\end{align*}
where the number $p^5 p^{(n-d)\ell/2}$
is independent of the choice of $(a,b)\in (A\times A)^{\leq\delta}$.
By Lemma~\ref{lem balance}, for $A\in\Omega_d$
we have $|(A\times A)^{\le\delta}|\le (p^\ell)^{2dh_q(\delta)}$.
So
\begin{align*}
& |{\cal D}^{\le\delta}| \leq \sum_{d=\mu(n)}^{n}~
\sum_{A\in\Omega_{d}}~|(A\times A)^{\le\delta}|\cdot
  p^5 p^{(n-d)\ell/2}\\
&  \le
 \sum_{d=\mu(n)}^{n}\sum_{A\in\Omega_{d}}
  p^{2 d \ell h_q(\delta)} p^5 p^{(n-d)\ell/2}
% &
 =  p^5 \! \sum_{d=\mu(n)}^{n}\sum_{A\in\Omega_{d}}
 p^{\frac{n\ell}{2} + 2 d\ell\big( h_q(\delta)-\frac{1}{4}\big)},
\end{align*}
where the number $p^{\frac{n\ell}{2}
 + 2 d\ell\big( h_q(\delta)-\frac{1}{4}\big)}$
is independent of the choice of $A\in\Omega_d$.
Using Eq.\eqref{eq Omega} yields
\begin{align*}
 &|{\cal D}^{\le\delta}|
 \le p^5\!\sum_{d=\mu(n)}^{n}\! |\Omega_{d}|\!\cdot\!
 p^{\frac{n \ell}{2} + 2 d\ell\big( h_q(\delta)-\frac{1}{4}\big)}
\le p^5\!
 \sum_{d=\mu(n)}^{n} \! n^{\frac{d}{\mu(n)}}
 p^{\frac{n\ell}{2} + 2 d\ell\big( h_q(\delta)-\frac{1}{4}\big)}
\\
&~ = p^5\! \sum_{d=\mu(n)}^{n}\!
 p^{\frac{n\ell}{2} + 2 d\ell\big( h_q(\delta)-\frac{1}{4}\big)
 + \frac{d\log_p(n)}{\mu(n)}  }
= p^5 \!\sum_{d=\mu(n)}^{n}
 p^{\frac{n\ell}{2}
 - 2 d\ell\big(\frac{1}{4} -h_q(\delta) -\frac{\log_p(n)}{2\mu(n)\ell}\big) }.
\end{align*}Note that $\frac{1}{4} -h_q(\delta) -\frac{\log_p(n)}{2\mu(n)\ell}>0$
(since $\frac{1}{4} -h_q(\delta) -\frac{\log_p(n)}{\mu(n)\ell}>0$)
and $d\ge\mu(n)$. We further get
\begin{align*}
|{\cal D}^{\le\delta}|
&\le p^5 \sum_{d=\mu(n)}^{n}
 p^{\frac{n\ell}{2} - 2\mu(n)\ell\big(\frac{1}{4}
 -h_q(\delta)-\frac{\log_p(n)}{2\mu(n)\ell}\big) }\\
&\le p^5\cdot n\cdot
 p^{\frac{n\ell}{2} - 2\mu(n)\ell\big(\frac{1}{4} -h_q(\delta)
 -\frac{\log_p(n)}{2\mu(n)\ell}\big) }.
\end{align*}
%where the last inequality holds since $n-d+1\leq n=p^{\log_{p}n}$.
That is,
$|{\cal D}^{\le\delta}|\leq p^5
 p^{\frac{n\ell}{2} - 2 \mu(n)\ell \big(\frac{1}{4} -h_q(\delta)
 -\frac{\log_p(n)}{\mu(n)\ell}\big) }$.
\end{proof}

By \cite[Lemma 2.6]{BM} (or \cite[Lemma II.6]{FL20}),
there are odd positive integers $n_1,n_2,\dots$ coprime to $q=p^\ell$ such that
$\lim\limits_{i\to\infty}\frac{\log_q(n_i)}{\mu(n_i)}=0$, where
$\mu(n)$ is defined in Eq.\eqref{eq mu(n)=...}.
Since $\log_q(n_i)=\log_p(n_i)/\log_p(q)=\log_p(n_i)/\ell$,
we see that
there are odd positive integers $n_1,n_2,\dots$ coprime to $p$ such that
\begin{align} \label{eq n_1 n_2 ...}
\lim\limits_{i\to\infty}\frac{\log_p(n_i)}{\mu(n_i)}=0,
\end{align}
which implies that $\mu(n_i)\to\infty$, hence $n_i\to\infty$.
Obviously, we can assume that $\mu(n_i)\geq \log_p(n_i)$ for $i=1,2,\dots$.

\begin{theorem} \label{thm H self-dual good}
Assume that $0<\delta<1-q^{-1}$ and $h_q(\delta)<\frac{1}{4}$
as in Eq.\eqref{eq delta in...}. Then there are
Hermitian self-dual $2$-quasi-cyclic codes $C_1, C_2, \dots$ over $F$
(hence ${\rm R}(C_i)=\frac{1}{2}$)
such that the code length $2n_i$ of $C_i$ goes to infinity and the
relative minimum distance $\Delta(C_i)>\delta$ for $i=1,2,\dots$.
\end{theorem}
\begin{proof}
Take $n_1,n_2,\dots$ as in Eq.\eqref{eq n_1 n_2 ...}.
There is a positive real number $\varepsilon>0$ such that
$\frac{1}{4}-h_q(\delta)-\frac{\log_p(n_i)}{\mu(n_i)\ell}\geq\varepsilon$
for large enough index $i$. So we can further assume that
\begin{align} \label{eq 1/4 -h_q...}
\frac{1}{4}-h_q(\delta)-\frac{\log_p(n_i)}{\mu(n_i)\ell}\geq\varepsilon,~~~~
i=1,2,\dots.
\end{align}
Taking $n=n_i$ in Corollary~\ref{cor |D|<} and Lemma~\ref{lem D^leq},
and denoting ${\cal D}$ by ${\cal D}(i)$,
we get
\begin{align*}
\frac{|{\cal D}(i)^{\le\delta}|}{|{\cal D}(i)|}\leq
\frac{p^5 p^{\frac{n_i \ell}{2} - 2 \mu(n_i) \ell
  \big(\frac{1}{4} -h_q(\delta)-\frac{\log_p(n_i)}{\mu(n_i)\ell}\big) }}
{p^{-2}p^{\frac{n_i\ell}{2}}}
= p^7 p^{-2 \mu(n_i) \ell
 \big(\frac{1}{4} -h_q(\delta)-\frac{\log_p(n_i)}{\mu(n_i)\ell}\big)}.
\end{align*}
By Eq.\eqref{eq 1/4 -h_q...} and that $\mu(n_i)\to\infty$, we get that
\begin{align*}
\lim_{i\to\infty} |{\cal D}(i)^{\le\delta}|\big/|{\cal D}(i)|
\leq \lim_{i\to\infty} p^7 p^{-2\mu(n_i) \ell \varepsilon} =0.
\end{align*}
Therefore, we can further assume that
$|{\cal D}(i)^{\le\delta}|<|{\cal D}(i)|$ for $i=1,2,\dots$.
So we can take $g_i\in {\cal D}(i)\setminus{\cal D}(i)^{\le\delta}$.
Then $C_i=C_{1,g_i}$ is a Hermitian self-dual $2$-quasi-cyclic code
of length $2n_i$ with $\Delta(C_i)>\delta$.
\end{proof}

\section{Hermitian (Euclidean) self-dual $2$-quasi\\ constacyclic codes}
\label{section H self-dual 2-QQ cyclic}

In this section we prove that if $\lambda^{1+p^{\ell/2}}=1$
($\lambda^{2}=1$ and $q\,{\not\equiv}\,3~({\rm mod}~4)$, respectively),
then Hermitian self-dual (Euclidean self-dual, respectively)
$2$-quasi $\lambda$-constacyclic codes are asymptotically good.
We first relate $\RL=F[X]/\langle X^n-\lambda\rangle$
with $\R1 =F[X]/\langle X^n-1\rangle$,
and then turn to the Hermitian case ($h=\ell/2$)
and the Euclidean case ($h=0$).

\begin{lemma} \label{lem n coprime to t}
Assume that ${\lambda^{1+p^h}=1}$ and $\gcd(n,t)=1$,
where $t={\rm ord}_{F^\times}(\lambda)$ as in Eq.\eqref{eq t=...}.
Then there is a $\gamma=\lambda^s\in F^\times$ such that
$\gamma^n=\lambda^{-1}$ and the map
\begin{align} \label{eq n coprime to t}
\eta:~ \R1\,\longrightarrow\, \RL,
\quad \sum_{i=0}^{n-1} a_iX^i\,\longmapsto\,
\sum_{i=0}^{n-1} a_i(\gamma X)^i~({\rm mod}~X^n-\lambda),
\end{align}
satisfies the following:

{\bf(1)} $\eta$ is an algebra isomorphism.

{\bf(2)}
The weight ${\rm w}\big(\eta(a(X))\big)={\rm w}\big(a(X)\big)$,
\, $\forall$ $a(X)\in\R1$.

{\bf(3)}
$\big\langle \eta(a(X)),\, \eta(b(X))\big\rangle_h=
\big\langle a(X),\, b(X)\big\rangle_h$, \, $\forall$ $a(X), b(X)\in\R1$.
\end{lemma}

\begin{proof}
Since $\gcd(n,t)=1$, there are integers $s,s'$ such that $ns+ts'=-1$.
We deduce that
$\lambda^{-1}=\lambda^{ns+ts'}=\lambda^{ns}\lambda^{ts'}=\lambda^{ns}$
since $\lambda^t=1$.
Take $\gamma=\lambda^s$, then $\gamma^{n}= \lambda^{-1}$.
It is known that (1) and (2) has been proved in \cite[Theorem 3.2]{CFLL}.
Therefore, it suffices to show that (3).
%The following is a surjective algebra homomorphism:
%$$
%\tilde\eta: F[X]\,\longrightarrow\, F[X]/\langle X^n-\lambda \rangle,
%\quad f(X)\,\longmapsto\, f(\gamma X) ~ ({\rm mod}~ X^n-\lambda),
%$$
%and
%$$
% \tilde\eta(X^n-1)=\gamma^n X^n -1=\lambda^{-1}X^n-1
% =\lambda^{-1}(X^n-\lambda)\equiv 0~~ ({\rm mod}~X^n-\lambda).
%$$
%Thus the kernel of $\tilde\eta$ is the ideal $\langle X^n-1\rangle$, and
%the surjective homomorphism $\tilde\eta$ induces the isomorphism $\eta$
%in Eq.\eqref{eq n coprime to t}. So (1) is proved.

%For any $a(X)\!=\!\sum_{i=0}^{n-1}a_i X^i\in\R1$,
%by Eq.\eqref{eq n coprime to t} we get
%$\eta\big(a(X)\big)\!=\!\sum_{i=0}^{n-1}a_i\gamma^i X^i$.
%So we see that
% $a_i\ne 0$ if and only if $a_i\gamma^i\ne 0$,
%hence ${\rm w}\big(\eta(a(X))\big)={\rm w}\big(a(X)\big)$.
%Thus (2) holds.

Let $a(X)=\sum_{i=0}^{n-1}a_i X^i\in \R1$ and
$b(X)=\sum_{i=0}^{n-1}b_i X^i\in \R1$.
Then $\eta\big(a(X)\big)=\sum_{i=0}^{n-1}a_i\gamma^i X^i$ and
$\eta\big(b(X)\big)=\sum_{i=0}^{n-1}b_i\gamma^i X^i$.
By Eq.(\ref{eq G1 inner prod.}), we have that
$$
\big\langle \eta(a(X)), \eta(b(X))\big\rangle_h =
\sum_{i=0}^{n-1} (a_i\gamma^i)\cdot (b_i\gamma^i)^{p^h}
=\sum_{i=0}^{n-1} a_ib_i^{p^h}\cdot \gamma^{(1+p^h)i}.
$$
Applying the equality $\gamma=\lambda^s$
and the assumption $\lambda^{1+p^h}=1$, we get
$$
 {\gamma^{1+p^h}=(\lambda^s)^{1+p^h}=(\lambda^{1+p^h})^s}=1.
$$
Thus
$$
\big\langle \eta(a(X)), \eta(b(X))\big\rangle_h
=\sum_{i=0}^{n-1} a_ib_i^{p^h}=\big\langle a(X), b(X)\big\rangle_h,
$$
which proves (3).
\end{proof}

\begin{corollary} \label{cor n coprime to t}
Assume that ${\lambda^{1+p^h}=1}$ and $\gcd(n,t)=1$,
where $t={\rm ord}_{F^\times}(\lambda)$ as in Eq.\eqref{eq t=...}. Let
\begin{align} \label{2 eq n coprime to t}
\eta^{(2)}:~ \R1^2\,\longrightarrow\, \RL^2,~~
\big(a(X),a'(X)\big)\,\longmapsto\, \big(\eta(a(X)),\eta(a'(X))\big),
\end{align}
where $\eta$ is defined in Lemma~\ref{lem n coprime to t}.
Then the following hold.

{\bf(1)} $\eta^{(2)}$ is a module isomorphism.

{\bf(2)} ${\rm w}\big(\eta^{(2)}(a(X),a'(X))\big)
={\rm w}\big(a(X),a'(X)\big)$,~ $\forall\; (a(X),a'(X))\in\R1^2$.

{\bf(3)} $\!\big\langle \eta^{(2)}(a(X),a'(X)), \eta^{(2)}(b(X),b'(X))\big\rangle_h
\!=\!\big\langle (a(X),a'(X)),\! (b(X),b'(X))\big\rangle_h$,
~$\forall\;(a(X),a'(X)),(b(X),b'(X))\in \R1^2$.
\end{corollary}

\begin{proof}
{\rm (1)}~ For $\big(a(X),a'(X)\big)\in{\R1}^2$ and $f(X)\in{\R1}$,
by Lemma \ref{lem n coprime to t}(1),
\begin{align*}
&\eta^{(2)}\big(f(X)\cdot(a(X),a'(X))\big)
 =\big(\eta(f(X)a(X)),\eta(f(X)a'(X))\big)\\
&=\big(\eta(f(X))\eta(a(X)),\eta(f(X))\eta(a'(X))\big)
 =\eta(f(X))\cdot\eta^{(2)}\big(a(X),a'(X)\big).
\end{align*}
Thus Eq.\eqref{2 eq n coprime to t} is a module homomorphism.
Observe that by Lemma \ref{lem n coprime to t}(1), $\eta^{(2)}$ must be bijective,
hence it is a module isomorphism.

{\rm (2)}~  By Lemma \ref{lem n coprime to t}(2), we get
\begin{align*}
{\rm w}\big(\eta^{(2)}(a(X),a'(X))\big)
&={\rm w}\big(\eta(a(X)),\eta(a'(X))\big)
 ={\rm w}\big(\eta(a(X))\big)+{\rm w}\big(\eta(a'(X))\big) \\
& ={\rm w}\big(a(X)\big)+{\rm w}\big(a'(X)\big)
 ={\rm w}\big(a(X),a'(X)\big).
\end{align*}

{\rm (3)}~
For $(a(X), a'(X))\in{\R1^{2}}$, $(b(X), b'(X))\in{\R1^{2}}$,
where $a(X)\!=\sum_{i=0}^{n-1}a_iX^i$, $a'(X)\!=\sum_{i=0}^{n-1}a'_iX^i$ and
$b(X)=\sum_{i=0}^{n-1}b_iX^i$, $b'(X)=\sum_{i=0}^{n-1}b'_iX^i$,
it is obvious that
\begin{align*}
 \big\langle (a(X), a'(X)),(b(X), b'(X))\big\rangle_{h}
& =\sum_{i=0}^{n-1}a_i{b_i}^{p^h}+\sum_{i=0}^{n-1}a'_i{b'_i}^{p^h} \\
&=\big\langle a(X), b(X)\big\rangle_{h}+\big\langle a'(X), b'(X)\big\rangle_{h}.
\end{align*}
Then
\begin{align*}
&\big\langle \eta^{(2)}(a(X),a'(X)),\;\eta^{(2)}(b(X),b'(X))\big\rangle_h \\
=&\big\langle \big(\eta(a(X)),\eta(a'(X))\big),
  \;\big(\eta(b(X)),\eta(b'(X))\big)\big\rangle_h \\
=&\big\langle \eta(a(X)),\eta(b(X))\big\rangle_h +
 \big\langle \eta(a'(X)),\eta(b'(X))\big\rangle_h \\
=&\big\langle a(X), b(X)\big\rangle_h +
 \big\langle a'(X),b'(X)\big\rangle_h
  \qquad\mbox{(by Lemma \ref{lem n coprime to t}(3))}\\
=&\big\langle (a(X),a'(X)),\; (b(X),b'(X))\big\rangle_h.
\end{align*}
We are done.
\end{proof}

\begin{theorem} \label{thm n coprime to t}
Assume that ${\lambda^{1+p^h}=1}$ and $\gcd(n,t)=1$, where
$t={\rm ord}_{F^\times}\!(\lambda)$ as in~Eq.\eqref{eq t=...}.
Let $\eta^{(2)}$ be as in Eq.\eqref{2 eq n coprime to t}.
Then $C$ is an $\R1$-submodule of~$\R1^2$ if and only if
$\eta^{(2)}(C)$ is an $\RL$-submodule of $\RL^2$.
At that case, the following hold.

{\bf(1)} The rate ${\rm R}\big(\eta^{(2)}(C)\big)={\rm R}(C)$.

{\bf(2)} The relative minimum distance $\Delta\big(\eta^{(2)}(C)\big)=\Delta(C)$.

{\bf(3)} $\eta^{(2)}(C)$ is Galois $p^h$-self-dual
 if and only if $C$ is Galois $p^h$-self-dual.
\end{theorem}

\begin{proof}
By Corollary \ref{cor n coprime to t}(1),
the map $\eta^{(2)}$ in Eq.\eqref{2 eq n coprime to t}
is a module isomorphism.
So $C$ is an $\R1$-submodule of $\R1^2$ if and only if
$\eta^{(2)}(C)$ is an $\RL$-submodule of $\RL^2$.
Assume that it is this case.

{\rm (1)}~Since $\eta^{(2)}$ is an isomorphism,
$\dim_F\eta^{(2)}(C)=\dim_F C$, and so
$${\rm R}\big(\eta^{(2)}(C)\big)
=\frac{\dim_F \eta^{(2)}(C)}{2n}=\frac{\dim_F C}{2n}={\rm R}(C).$$

{\rm (2)}~ It holds by Corollary \ref{cor n coprime to t}(2) obviously.

{\rm (3)}~ Observe that by Corollary \ref{cor n coprime to t}(3),
 $\big\langle \eta^{(2)}(C),\eta^{(2)}(C)\big\rangle_h=0$ in $\RL^2$
if and only if $\big\langle C,\,C\big\rangle_h=0$ in $\R1^2$.
Further, applying the above conclusion (1) yields that
${\rm R}\big(\eta^{(2)}(C)\big)=\frac{1}{2}$ if and only if ${\rm R}(C)=\frac{1}{2}$.
Thus (3) holds.
\end{proof}

In the following, we consider the asymptotic property of
Hermitian (Euclidean) self-dual $2$-quasi
$\lambda$-constacyclic codes.

\begin{theorem} \label{thm H self-dual lambda good}
Assume that $\ell$ is even.
The Hermitian self-dual $2$-quasi $\lambda$-constacyclic codes over $F$
are asymptotically good if and only if $\lambda^{1+p^{\ell/2}}=1$.
\end{theorem}

\begin{proof}
If $\lambda^{1+p^{\ell/2}} \neq 1$,
then Hermitian self-dual $2$-quasi $\lambda$-constacyclic codes over~$F$
are asymptotically bad, see Theorem~\ref{thm case not good}.

Assume that $\lambda^{1+p^{\ell/2}}=1$.
Let $0<\delta<1-q^{-1}$ and $h_q(\delta)<\frac{1}{4}$.
By Theorem~\ref{thm H self-dual good}, there are
Hermitian self-dual $2$-quasi-cyclic codes $C_1,C_2, \dots$
over~$F$ such that:
\begin{itemize}
\item
 the code length $2n_i$ of $C_i$ satisfy that $n_i$ is odd and coprime to $q$
 and
$\lim\limits_{i\to\infty} \frac{\log_p(n_i)}{\mu(n_i)}=0$;
 in particular, $\mu(n_i)\to\infty$, and hence $n_i\to\infty$;
\item
 the rate ${\rm R}(C_i)=\frac{1}{2}$
and the relative minimum distance $\Delta(C_i)>\delta$\, for $i=1,2,\dots$.
\end{itemize}
Let $\tilde t=\max\big\{\mu(r)\,\big|\,
  \mbox{$r$ runs over the prime divisors of $t$}\big\}$,
where $t={\rm ord}_{F^\times}\!(\lambda)$ as in~Eq.\eqref{eq t=...}.
By \cite[Lemma II.2]{LF22},
\begin{align*} %\label{eq mu_q min p|n}
\mu(n_i)=\min\big\{\mu(r)\;|\;
  \mbox{$r$ runs over the prime divisors of $n_i$}\big\}.
\end{align*}
If $\gcd(n_i,t)\ne 1$, then $\mu(n_i)\le\tilde t$.
Since $\mu(n_i)\to\infty$,
there are only finitely many~$n_i$ such that $\gcd(n_i,t)\ne 1$.
Removing such $n_i$, we can further assume that
$\gcd(n_i,t)=1$ for $i=1,2,\dots$. Thus,
applying the isomorphism Eq.\eqref{2 eq n coprime to t} to~$C_i$
yields $\tilde C_i=\eta^{(2)}(C_i)$ in $\RL^2$, $i=1,2,\dots$.
We get the code sequence
$$
 \tilde C_1,~\tilde C_2,~\dots .
$$
By Theorem \ref{thm n coprime to t}, each $\tilde C_i$
is Hermitian self-dual $2$-quasi $\lambda$-constacyclic codes over $F$,
and the code length $2n_i$ goes to infinity,
the rate ${\rm R}(\tilde C_i)=\frac{1}{2}$
and the relative minimum distance
$\Delta(\tilde C_i)=\Delta(C_i)>\delta$ for $i=1,2,\dots$.
\end{proof}

For Euclidean case,
the ``Euclidean self-dual'' is referred to as ``self-dual''.

\begin{theorem} \label{thm E self-dual lambda good}
Assume that $q\;{\not\equiv}\;3~({\rm mod}~4)$.
The self-dual $2$-quasi $\lambda$-constacyclic codes over $F$
are asymptotically good if and only if $\lambda^2=1$.
\end{theorem}

\begin{proof}
By Theorem~\ref{thm case not good}, if $\lambda^2\neq 1$ then
the self-dual $2$-quasi $\lambda$-constacyclic codes are asymptotically bad.
In the following we assume that $\lambda=\pm 1$ .

If $\lambda=1$, it has been shown in \cite{LF22} (cf. Remark~\ref{rk 2-quasi cyclic})
that the self-dual $2$-quasi-cyclic codes over $F$ are asymptotically good.

Assume that $\lambda=-1$
(i.e., consider the $2$-quasi negacyclic codes).
Let $0<\delta<1-q^{-1}$ and $h_q(\delta)<\frac{1}{4}$.
Since $q\;{\not\equiv}\;3~({\rm mod}~4)$, by
Remark~\ref{rk 2-quasi cyclic}, there are
 self-dual $2$-quasi-cyclic codes $C_1,C_2, \dots$ over $F$, such that:
\begin{itemize}
\item
 the code length $2n_i$ of $C_i$ satisfy that $n_i$ is odd and coprime to $q$,
 and $\lim\limits_{i\to\infty} \frac{\log_p(n_i)}{\mu(n_i)}=0$;
in particular,  $\mu(n_i)\to\infty$, and hence $n_i\to\infty$;
\item
the rate ${\rm R}(C_i)=\frac{1}{2}$
and the relative minimum distance $\Delta(C_i)>\delta$\, for $i=1,2,\dots$.
\end{itemize}
Similarly to the proof of Theorem~\ref{thm H self-dual lambda good},
we get self-dual $2$-quasi negacyclic codes
$$\tilde C_1,~ \tilde C_2,~\dots $$
over $F$ such that their code length $2n_i$ goes to infinity,
${\rm R}(\tilde C_i)=\frac{1}{2}$
and their relative minimum distance $\Delta(\tilde C_i)=\Delta(C_i)>\delta$ for $i=1,2,\dots$.
\end{proof}

\section{Conclusions}\label{Conclusions}

The purpose of this paper is to characterize
the Galois self-dual $2$-quasi $\lambda$-constacyclic codes, and
to investigate their asymptotic properties.

We first showed the algebraic structure of $2$-quasi $\lambda$-constacyclic codes.
Then we found that the Galois $p^h$-self-dual $2$-quasi $\lambda$-constacyclic codes
behave much differently according to whether $\lambda^{1+p^h}\!\neq 1$
($\lambda\neq \lambda^{-p^{\ell-h}}$ equivalently) or
${\lambda^{1+p^h}\!=1}$
($\lambda=\lambda^{-p^{\ell-h}}$ equivalently).
In both the cases, we exhibited the necessary and sufficient conditions for the
 $2$-quasi $\lambda$-constacyclic codes being Galois $p^h$-self-dual,
 see Theorem~\ref{th lambda neq iff} (for the former case),
Lemma \ref{lem C1,g} and Theorem~\ref{th lambda eq iff} (for the latter case).
And in the former case we proved that
the Galois $p^h$-self-dual $2$-quasi $\lambda$-constacyclic codes
are asymptotically bad.

Then we focused on the case that $\lambda^{1+p^h}=1$.
A methodological innovation is that we introduced
(in Eq.\eqref{def *} and Lemma~\ref{lem tau and *})
the operator ``$*$'' on the algebra $F[X]/\langle X^n-\lambda\rangle$,
which is proved to be a powerful technique for studying Galois dualities.
An important contribution is that we proved that
the Hermitian self-dual (when $\ell$ is even) and the
Euclidean self-dual (when $q\,{\not\equiv}\,3~({\rm mod}~4)$)
$2$-quasi $\lambda$-constacyclic codes are asymptotically good.
The proof are divided into two steps.
First we proved (in Theorem~\ref{thm H self-dual good}) the asymptotic goodness of the
Hermitian self-dual $2$-quasi-cyclic codes
(the asymptotic goodness of the Euclidean self-dual $2$-quasi-cyclic codes
has been obtained in \cite{LF22}).
And then we relate the $2$-quasi $\lambda$-constacyclic codes
to the $2$-quasi-cyclic codes by an algebra isomorphism
which preserves the Hamming weight and the Galois $p^h$-inner products,
see Corollary~\ref{cor n coprime to t};
hence the asymptotic goodness of the Hermitian self-dual and Euclidean self-dual
$2$-quasi $\lambda$-constacyclic codes are derived
(Theorem~\ref{thm H self-dual lambda good}
and Theorem~\ref{thm E self-dual lambda good}).

%It is still an open question:
A question remains unsolved:
with the assumption that $\lambda^{1+p^h}=1$,
are the Galois $p^h$-self-dual $2$-quasi $\lambda$-constacyclic codes over $F$,
except for the Hermitian self-dual ones
and the Euclidean self-dual (when $q\,{\not\equiv}\,3~({\rm mod}~4)$) ones,
asymptotically good?
It seems that the existing approaches in this paper
are not enough to solve this question.
We look forward this question to be solved perfectly in the future.
A special sub-question is: are
the self-dual $2$-quasi negacyclic codes over $F$ asymptotically good?
%As mentioned in Introduction,
The result in \cite{SQS} and Theorem~\ref{thm E self-dual lambda good}
of this paper together imply a positive answer to this sub-question;
but the argument of~\cite{SQS} depends on Artin's primitive root conjecture.
Recently, % based on the methods of this paper,
in \cite{FL24} we further developed a number-theoretic and algebraic method
to analyse the $q$-cosets and proved the asymptotic goodness of
any $q$-ary self-dual $2$-quasi negacyclic codes.

\section*{Acknowledgements}
%NSFC for the supports through Grant 11271005.
The authors are grateful to the editor and the anonymous referees
for taking time to read and comment on the article very carefully and insightfully.
Their nice comments and suggestions helped them to improve the article very much.

\end{document}